\documentclass[12pt]{article} 
\pdfoutput=1
\usepackage[authoryear]{natbib}
\usepackage{graphicx,epstopdf}
\usepackage{setspace}
\usepackage[margin=2.5cm]{geometry}
\usepackage{array}
\usepackage{hyperref}

\usepackage{pdflscape}
\usepackage{lscape}
\usepackage{graphicx}
\usepackage[utf8]{inputenc}
\usepackage[english]{babel}
\usepackage{amsmath, amsthm, amssymb, amsfonts}
\newtheorem{theorem}{Theorem}[section]

\newtheorem{lemma}[theorem]{Lemma}
\theoremstyle{definition}
\newtheorem{definition}{Definition}[section]
\newtheorem{example}{Example}[section]
\newtheorem{assumption}{Assumption}[section]

\usepackage{float}
\floatstyle{ruled}
\newfloat{algorithm}{tbp}{loa}
\providecommand{\algorithmname}{Algorithm}
\floatname{algorithm}{\protect\algorithmname}

\usepackage{ifpdf}
\ifpdf
  \DeclareGraphicsExtensions{.pdf,.png,.jpg}
\else
  \DeclareGraphicsExtensions{.eps}
\fi

\usepackage{natbib}
\setcitestyle{authoryear,open={(},close={)}}

\usepackage{enumitem}
\usepackage{subfig}
\usepackage[section]{placeins}
\usepackage{framed}

\begin{document}

\title{An optimal test for strategic interaction \\ in social and economic network formation \\ between heterogeneous agents}
\author{Andrin Pelican and Bryan S. Graham
\thanks{\underline{Pelican}: e-mail: [pelicanandrin@gmail.com]. 
\underline{Graham}: Department of Economics, University of California - Berkeley, 530 Evans Hall \#3380, Berkeley, CA 94720-3880 and National Bureau of Economic Research, e-mail: [bgraham@econ.berkeley.edu], web: \texttt{http://bryangraham.github.io/econometrics/}. We thank seminar participants at the University of California - Berkeley, the University of Cambridge, ITAM, Jinan University, Yale University, Microsoft Research, Princeton University, Rice University, Boston University, the World Congress of the Econometric Society as well as the North American Winter Meetings of the Econometric Society for feedback and suggestions. We also thank Daniele Ballinari, Vincent Boucher, Shuowen Chen, Enrico De Giorgi, Aureo de Paula, P\'eter Erd\"os, Bo Honor\'e, Hiroaki Kaido, Piotr Lukaszuk, Ulrich Mueller, Kaïla A. Munro, Gabriel Okasa, Seth Richards-Shubik, Xun Tang and Yi Zhang for their generous feedback and input on earlier drafts. A special thanks to Winfried Hochstättler for continuously pointing out the connection to the discrete mathematics literature and to Michael Jannson for help and insight into the nature of our testing problem. This revision has also benefited from discussions about the theory of strategic network formation with Sanjeev Goyal and Matt Jackson, very detailed comments from the co-editor, as well as the feedback and suggestions provided by four anonymous referees. All the usual disclaimers apply. Financial support from NSF grant SES \#1357499 is gratefully acknowledged by the second author. This paper is a revised version of the co-authored chapter ``Testing Strategic Interaction in Networks" which appeared in Pelican's 2019 doctoral dissertation at the University of St. Gallen.}} 

\date{\today}

\maketitle

\newpage

\begin{abstract}
Consider a setting where $N$ players, partitioned into $K$ observable types, form a directed network. Agents' preferences over the form of the network consist of an arbitrary network benefit function (e.g., agents may have preferences over their network centrality) and a private component which is additively separable in own links. This latter component allows for unobserved heterogeneity in the costs of sending and receiving links across agents (respectively out- and in- degree heterogeneity) as well as homophily/heterophily across the $K$ types of agents. In contrast, the network benefit function allows agents' preferences over links to vary with the presence or absence of links elsewhere in the network (and hence with the link formation behavior of their peers). In the null model which excludes the network benefit function, links form independently across dyads in the manner described by \cite{Charbonneau_EJ17}. Under the alternative there is interdependence across linking decisions (i.e., strategic interaction). We show how to test the null with power optimized in specific directions. These alternative directions include many common models of strategic network formation (e.g., ``connections" models, ``structural hole" models etc.). Our random utility specification induces an exponential family structure under the null which we exploit to construct a similar test which exactly controls size (despite the the null being a composite one with many nuisance parameters). We further show how to construct locally best tests for specific alternatives without making any assumptions about equilibrium selection. To make our tests feasible we introduce a new MCMC algorithm for simulating the null distributions of our test statistics. 

\end{abstract}



\newpage

\pagenumbering{arabic}
\onehalfspacing

\vspace{5mm}

In an \emph{economic} model of (directed) network formation agents \emph{purposefully} direct links to one another in order to maximize utility. Specifically, a payoff function maps all possible network configurations into agent utilities. Agents use this payoff function to weigh the benefits of directing any particular link against the costs of doing so. A Nash Equilibrium (NE) network arises when all agents link choices are individually optimal given the choices made by other agents \citep[e.g.][]{Bala_Goyal_EM00}.

Important examples of such processes include firms choosing their suppliers \citep[e.g.,][]{Atalay_et_al_PNAS11}, adolescents choosing friends \citep[e.g.,][]{Christakis_et_al_Book2020}, banks engaging in interbank lending to meet statutory reserve requirements \citep[e.g.,][]{Boss_et_al_QF2004}, and village households choosing partners for risk-sharing \citep[e.g.,][]{deWeerdt_IAP04}. \cite{Jackson_et_al_JEL17} present many other examples of networks in economics. Such data abound in the other social sciences as well \citep[e.g.,][]{Apicella_et_al_Nat12}.

The utility an agent receives when she directs a link to another agent can be usefully divided into two components.\footnote{In digraphs, or directed networks, it is customary to refer to edges as ``arcs". Here we use the terms link, edge, arc, friendship, relationship etc. interchangeably.} The first component is ``private", or, more precisely, invariant to the presence or absence of other links in the network.\footnote{Note ``other links" include those possibility directed by the sending agent to targets other than the one at hand. An alternative to the ``private" nomenclature would be ``dyadic" or ``direct".} The second component is ``social", or varying with the presence or absence of other links in the network.

An example of the first component is the payoff associated with a homophilous link \citep{McPherson_et_al_ARS01}. This payoff component only depends on the attributes of the sending (ego) and receiving (alter) agents. Another example is associated with ``degree heterogeneity": agents may vary systematically in their propensity to direct links, or in their attractiveness as link targets for others. Finally we might posit that the payoff from any particular link varies for idiosyncratic reasons, as in other random utility models (RUMs) of discrete choice \citep{McFadden_FinE74}. Empirical models of network formation with these features were introduced by \cite{Charbonneau_EJ17}, \cite{Graham_EM17}, \cite{Dzemski_RESTAT18}, \cite{Jochmans_JBES18} and \cite{Yan_et_al_JASA18}. These models are fundamentally \emph{dyadic}: agents' network payoffs are a simple sum of link-specific payoffs and, crucially, invariant to the linking behavior of other agents.

In some settings, however, agents may also value indirect links. For example, an arc from $j$ to $k$ may incidentally reduce the shortest path length from $i$ to $k$, allowing agent $i$ better access to $k$'s information \citep[e.g.,][]{Jackson_Wolinsky_JET96, Bala_Goyal_EM00}. While arc $jk$ is valued by $i$, this value is not incorporated into $j$'s decision to direct the arc or not. Preferences of this type mean agents' decisions impose \emph{externalities} on others. The detection of such externalities is the subject of this paper.

Payoff functions with externalities feature prominently in formal theoretical models of network formation \citep[cf.,][]{Jackson_NetBook08, Goyal_Book2021}. Equilibrium in network formation models with externalities may be analyzed using the tools of game theory. Indeed such models are typically called strategic network formation models. In what follows we say a network formation model is \emph{strategic} if agents value indirect links or, equivalently, their optimal linking strategy varies with the linking behavior of others.

When links made by one agent alter the incentives for link formation faced by others, equilibrium network configurations may diverge from socially optimal ones \citep{Goyal_Book2021}. This, in turn, suggests that well-designed interventions might make agents better off. In contrast, without a wedge between the private and social benefits of link formation, equilibrium and socially optimal networks will coincide. This paper introduces a test for whether agents' own incentives to form links vary with the choices of others. A rejection of our test, under the maintained model, indicates the presence of externalities, with their attendant implications for optimal policy design.

\subsubsection*{An overview of the test and its uses}
Strategic network formation games are complicated. In a directed network with $N$ agents, there are $2^{N(N-1)}$ possible action profiles or network configurations; many of which may be Nash Equilibria (NE). In the seminal model of directed network formation introduced by \cite{Bala_Goyal_EM00}, for example, with $N=5$ agents there are $1,069$ NE networks. Because of this combinatoric complexity, methods pioneered for the econometric analysis of discrete games with just a few players are not directly applicable -- at least in practice -- to network formation games.

In recent work, \cite{Christakis_et_al_Book2020}, \cite{Mele_EM17}, \cite{Miyauchi_JOE16}, \cite{dePaula_et_al_EM18} and \cite{Sheng_EM20} each proposed empirical models of strategic network formation.\footnote{\cite{dePaula_ARE2020} surveys work in this area and provides additional references.} Each of these models impose particular restrictions on the form of the network payoff function, the nature of any unobserved heterogeneity, and/or make assumptions about equilibrium selection. Even with these restrictions, estimating the identified set for the parameters indexing the network payoff function in these models is challenging, as is conducting inference.\footnote{We wish to emphasize that these ``critiques" reflect the inherent difficulty of the problem, not any deficiencies in the above cited papers. Indeed these researchers have shown considerable ingenuity in proposing ways to make methods designed for games with just a few players scale to the considerably more complicated many-player network setting.}

In this paper we introduce an econometric model of strategic network formation which, we believe for the first time, simultaneously allows (i) for agents to value both direct and indirect links, (ii) for the systematic returns to link formation to vary with observed dyad attributes, and (iii) for unobserved agent-specific correlated degree heterogeneity. Our setup maps neatly into the ``costs versus benefits" payoff structures emphasized in theoretical models of strategic network formation (see, for example, \citet[Chapters 6 \& 11]{Jackson_NetBook08} and \citet[Chapter 3]{Goyal_Book2021}). Examples of models -- suitably enriched to include covariates, unobserved heterogeneity, and random link utility -- encompassed by our framework include the ``connections" model \citep[e.g.,][]{Jackson_Wolinsky_JET96, Bala_Goyal_EM00}, ``structural hole" or ``bridging" models \citep[e.g.,][]{Goyal_Vega-Redondo_JET2007, Kleinberg_et_al_ACM2008} and the favor exchange or ``supported links" model of \cite{Jackson_et_al_AER12}. We can also accommodate tastes for reciprocity, transitivity, network centrality and other forms of indirect link valuation.

We begin with the baseline dyadic logistic regression model for directed networks introduced by \cite{Charbonneau_EJ17}.\footnote{Although the dissertation from which \cite{Charbonneau_EJ17} was drawn appears to be the first formal analysis of the dyadic logit model (especially in terms of exploring the implications of its exponential family structure for estimation), its use in empirical work arose earlier. For example, in the empirical network analysis of \cite{deWeerdt_IAP04}; see also \cite{Holland_Leinhardt_JASA81}.} This model is useful for modelling homophily and degree heterogeneity. We then augment this model with a network payoff term which additionally allows agents to value indirect links. The resulting model is quite complicated. Formally it is a very large complete information simultaneous move game. While we assume that the observed network is a NE, we make no auxiliary equilibrium selection assumptions.\footnote{More precisely the observed network is either a pure strategy NE or in the support of a mixed strategy NE (in fact our results hold under an even weaker notion of equilibrium, as explained below).}

Let $K$ be the number of support points in the distribution of observed agent attributes and $N$ the number of agents in the network. Our model includes (i) $K^2$ ``homophily"  parameters, $\Lambda\overset{def}{\equiv}\left[\lambda_{kl}\right]$ for $k,l=1 \dots K$, capturing how link returns vary systematically with ego and alter attributes, (ii) two $N \times 1$ parameter vectors $\mathbf{A}\overset{def}{\equiv}[A_i]$ and $\mathbf{B}\overset{def}{\equiv}[B_i]$ for $i=1 \dots N$, capturing, respectively, agent-specific out- and in-degree heterogeneity, and (iii) a scalar parameter, $\gamma$, measuring the extent to which agents value indirect links. Our model also includes (iv) an ``equilibrium selection" function. Since we are agnostic about which NE is selected in the presence of multiple equilibria, this function is not specified by the analyst, but enters our analysis abstractly (see Theorem \ref{thm: measurability} below).

We treat $\delta=(\Lambda',\bf{A}',\bf{B}')'$ as a (high dimensional) nuisance parameter and the equilibrium selection mechanism as a nuisance function.\footnote{This function assigns probabilities to all NE equilibria for every possible realization of the random utility shocks.} This focuses our attention solely on $\gamma$. While, in principle, an analysis of the identified set for $\gamma$ might be possible, we instead focus on the one-sided hypothesis of $H_0:\gamma=0$ versus $H_1:\gamma>0$. Or, put differently, we identify the \emph{sign} of $\gamma$.\footnote{Our focus on one-sided hypotheses results in a particularly clean exposition and analysis, allows for the statement of some optimality results, and covers our main examples of interest. However, as will be apparent, our basic set-up extends naturally to two-sided hypotheses.}

Our test involves comparing a statistic of the observed network (e.g., its transitivity index) with a critical value derived from a reference distribution. Natural questions are: (i) which reference distribution? (ii) how do I compute the critical value? (iii) which network statistic should I use? We provide answers to all three of these questions.

There is a long tradition in empirical work of using the \"{E}rdos-R\'{e}nyi model to generate the reference distribution. This invariably results in ``straw man" tests since few real world networks are well described by the \"{E}rdos-R\'{e}nyi model. To avoid spurious rejection of the null of no strategic interaction ($H_0:\gamma=0$) it is therefore important to have a rich null model; one that might actually describe real world networks. \cite{Charbonneau_EJ17} provides an easy to interpret, random utility based, and ``credible" null model.\footnote{The \cite{Charbonneau_EJ17} model can match any observed in- and out-degree sequence as well as rich patterns of homophilous linking. This is important since heavy-tailed degree distributions characterize many real world networks, as does homophily \citep[e.g.,][]{Barabasi_Book16, McPherson_et_al_ARS01}.},\footnote{An analogy: consider the challenge of determining whether persistence in panel data is due to state-dependence or unobserved heterogeneity (or both). Any credible test for state dependence needs to include as part of its null a correct specification of unobserved agent-specific heterogeneity \citep[e.g.,][]{Chamberlain_LALMD85}.} 

Because $\delta$ may range freely across its parameter space when $\gamma=0$ our null hypothesis is a composite one. Test size equals the supremum of the rejection rate across all data generating processes (DGPs) with $\gamma=0$. Because $\delta$ is high dimensional, the null space is large and constructing a test with good size and power properties is non-trivial  \citep[cf.,][]{Moreira_JOE09}. An additional non-standard feature of our testing problem is that the nuisance equilibrium selection function is only present under the alternative \citep[cf.,][]{Andrews_Ploberger_EM1994}. 

Under a logistic assumption on the random component of link utility, using a classic exponential family conditioning argument, we introduce a family of similar tests. We provide an \emph{exact} characterization of the null distributions of the test statistics in this family and, crucially, a feasible Markov Chain Monte Carlo (MCMC) algorithm for simulating from them. Simulating the null distribution requires drawing a binary adjacency matrix uniformly at random from the set of all adjacency matrices satisfying certain constraints. Constrained binary matrix simulation has numerous applications in biology, psychology, ecology and other fields \citep[cf.,][]{Sinclair_Book1993, Blitzstein_Diaconis_IM11}. Unfortunately, extant simulation algorithms cannot be used to simulate the null distribution needed here; our algorithm is therefore novel and of independent interest.

We also derive the form of the locally best test under the alternative $H_1:\gamma>0$. Remarkably we are able to do this while remaining agnostic about equilibrium selection. Finally, because our test is exact, we also side-step difficult issues that arise when undertaking asymptotic analysis in the single network context (see \cite{Graham_HBE20} for references and discussion).

Possible use cases for the methods introduced in this paper include:
\begin{enumerate}
    \item \textbf{Assessing model adequacy or goodness-of-fit:} The researcher believes the null model of \cite{Charbonneau_EJ17} is adequate for the setting at hand, but wishes to report an omnibus goodness-of-fit test (similar to the practice of reporting the Sargan-Hansen J-Statistic in the context of GMM estimation). While a rejection in this setting is interpreted as evidence against the baseline null model, it not interpreted as evidence in favor of any particular alternative.\footnote{Dyadic regression analysis has a long history in economics going back, at least, to the work of \cite{Tinbergen_SWE62}. See \cite{Graham_HBE20} for a survey and references. We note that this use case has the potential to introduce pre-testing bias if researchers only report their results conditional on accepting the null.} Our use of ``classic" sufficiency arguments separates the the information in the data relevant for estimation of $\delta$ -- the model parameter under the null -- from that relevant for assessing model adequacy \citep[cf.,][p. 29]{Barndorff-Nielsen_Cox_IA1994}. As is well-known, it is not possible to construct a test with good power in all directions of mis-specification \citep[][Theorem 14.6.2]{Lehmann_Romano_TSH05}. The researcher's choice of test statistic should therefore, at least heuristically, reflect those directions of mis-specification of most concern. 
   \item \textbf{Detecting strategic interaction of a specific form:} The researcher's primary interest is in the specified model and she wishes to sign identify $\gamma$. In this example the analysts undertakes empirical work under the maintained assumption that \emph{the true model is either in the null model space or in the specified alternative model space.} The data are used to determine which case prevails. This knowledge is actionable. For example, knowledge that $\gamma>0$ may be sufficient to justify policies which subsidize link formation. 
   \item \textbf{Cataloging ``unusual" network features:} The researcher wishes to assess whether certain features of the network in hand are ``unusual". In contrast to the first use case, here the researcher suspects that the network in hand is not well-described by the \cite{Charbonneau_EJ17} null, but, in contrast to the second use case, she remains somewhat agnostic about the form of the true model. The null model defines a set of reference networks with certain properties identical to those in the network of interest (e.g., the in- and out- degree sequences, numbers of links between agents with different covariate configurations). The researcher can compare features of interest in their network (e.g., diameter, reciprocity, support) with their distributions across the null reference set to assess whether their network is, indeed, ``unusual". There is a long history, as noted earlier, of comparing network statistics to their expected value under an \"{E}rdos and R\'{e}nyi null. Here we provide a more realistic reference null distribution. See \cite{Holland_Leinhardt_SM76} for a discussion of this type of analysis in sociology, Section 5 of \cite{Jackson_et_al_AER12} for an example from economics; \cite{Milo_el_al_Sci02} for an example from computational biology, and \cite{Gotelli_EC2000} for a discussion of applications to species co-occurrence analysis in ecology. Researchers undertaking this last type of analysis might be best described as doing structured data exploration.
    
\end{enumerate}

While our focus is on strategic network formation, it seems likely that the ideas developed below could be adapted to design tests appropriate for other incomplete econometric models. In recent work \cite{Chen_et_al_EM2018} and \cite{Kaido_Zhang_arXiv2019} introduced likelihood ratio type tests applicable to incomplete models. Our test, in contrast, is a conditional score test. Conditioning, while requiring exponential family structure, is helpful in settings with a high dimensional nuisance parameter \citep[cf.,][]{Moreira_JOE09}. Our score-based approach may also have computational advantages in settings where likelihood evaluation under the alternative is difficult (e.g., when enumeration of all NE is impossible).

\subsubsection*{Outline of the paper}

Section \ref{sec:model} presents our model of strategic network formation. We begin by defining agent preferences and characterizing equilibrium networks. With this foundation we are able to write down a likelihood function for the network. Since there may exist multiple equilibrium networks, this likelihood depends on an unknown (and unmodelled) equilibrium selection mechanism. Although well-defined (see Theorem \ref{thm: measurability} below), our likelihood function cannot be numerically evaluated in practice.

Section \ref{sec:test} outlines our approach to testing. We first characterize the exact distribution of any statistic of the adjacency matrix under the null. By conditioning on a sufficient statistic for the parameter of the null model we guarantee similarity of our text. Our test exactly controls size across all null model parameter values. Next we derive the form of the locally best test statistic for specific alternatives. 

Although we characterize the exact null distribution of our test statistics, for reasons of practically, we approximate this distribution by simulation. Section \ref{sec:simulation} outlines our new Markov Chain Monte Carlo (MCMC) algorithm for generating random draws from the required null distribution. 

Section \ref{sec:application} illustrates our methods in the context of the Nyakatoke risk-sharing network studied by \cite{deWeerdt_IAP04} and others. We construct a test with power for an alternative where agents value their ability to broker transactions among otherwise disconnected agents (as in Burt's \citeyear{Burt_StructuralHoles1995} theory of ``structural holes"). We formalize this idea using the model of \cite{Kleinberg_et_al_ACM2008}. This example is interesting because the form of a good test statistic is \emph{ex ante} non-obvious, but flows naturally from the \cite{Kleinberg_et_al_ACM2008} model and our results. The example also illustrates how test statistics need not be simple functions of the adjacency matrix or even exist in closed form. An extensively narrated Python Jupyter Notebook replicating this empirical illustration is available as part of the Supplemental Materials.

Section \ref{sec:extensions} finishes with a short discussion of limitations of our methods as well as a few thoughts on possible areas for additional research.

Proofs as well as some Monte Carlo simulation results are collected in a Supplemental Web Appendix. This appendix also includes a discussion of some additional applications of our MCMC simulation algorithm.

Readers interested primarily in applications can read Section \ref{sec:model}, the first part of Section \ref{sec:test}, and the empirical illustration of Section \ref{sec:application}. The balance of the paper can be read later (perhaps after viewing the Python Jupyter Notebook available in the supplemental materials).

\section{An family of empirical models of strategic network formation}\label{sec:model}

\subsection{Notation and setup}
A directed graph $G(\mathcal{V},\mathcal{A})$ consists of a set of vertices (agents) $\mathcal{V}=\{ 1,\ldots,N\}$ and a set of ordered pairs of nodes, respectively called \emph{egos} and \emph{alters}, $\mathcal{A}=\{( i,j) , (k,l) ,\ldots\}$ for $i\neq j$, $k\neq l$, and $i,j,k,l\in\mathcal{N}$. The elements of $\mathcal{A}$ correspond to those arcs, or directed links, present in $G(\mathcal{V},\mathcal{A})$. 

In what follows we typically work with the adjacency matrix $\mathbf{D}=[D_{ij}]$ where
\begin{equation}
D_{ij}=\left\{ 
\begin{array}{cc}
1 & \mbox{if}\,\, ij \in\mathcal{A}\\
0 & \mbox{otherwise}
\end{array}
\right. .
\end{equation}
Since we rule out self-links, the diagonal of $\mathbf{D}$ consists of structural zeros.

Let $G-ij$ denote the network obtained by deleting link $ij$ from $G$ (if present), and $G+ij$ the network one gets after adding this link (if absent). Let $\mathbf{D}\pm ij$ denote the adjacency matrix associated with the network obtained by adding/deleting link $ij$ from $G$. 

The set of all $2^{N(N-1)}$ possible adjacency matrices is denoted by $\mathbb{D}_{N}$. Hence $\mathbf{d}\in\mathbb{D}_N$ is a feasible network wiring or, equivalently, a pure strategy profile. Let $\mathbf{d}_{i}$ be the $i^{th}$ row of $\mathbf{d}$, or the pure strategy selection of agent $i$ (i.e., a binary vector indicating which edges she chooses to direct). The pure strategy profile for all players other than $i$ is denoted by $\mathbf{d}_{-i}$. We will sometimes refer to ``players other than $i$" as $i$'s \emph{peers}.

For each agent there are $M\overset{def}{\equiv}2^{N-1}$ possible actions, corresponding to all possible configurations of links she may direct towards her peers. A mixed strategy for agent $i$, $\sigma_{i}=\left(\pi_{1i},\pi_{2i},\ldots,\pi_{Mi}\right)'$, is probability distribution on these $M$ pure strategies; $\sigma=\left(\sigma_{1},\sigma_{2},\ldots,\sigma_{N}\right)'$ is a mixed strategy profile for all $N$ agents, while $\sigma_{-i}$ is the strategy profile of agent $i$'s peers.

\subsection{Preferences}

 The utility or payoff agent $i$ gets from network $\mathbf{d}$ is
\begin{equation} \label{eq: payoff_function}
\nu_{i}\left(\mathbf{d}_{i},\mathbf{d}_{-i};\theta, \mathbf{U}_i\right)=\underset{\text{Network Benefit}}{\underbrace{\gamma_{0}g_{i}\left(\mathbf{d}\right)}} - \underset{\text{Link Costs}}{\underbrace{\sum_{j}d_{ij}c_{ij}\left(X_i,X_j;\delta, U_{ij}\right)}}\end{equation}
with $g_{i}\left(\mathbf{d}\right)$ a known, but not necessarily closed-form, function of the network adjacency matrix, normalized such that $g_{i}\left(\mathbf{0}\right)=0$, $\theta=\left(\gamma,\delta'\right)'$, and the link ``costs'' function taking the form
\begin{equation}
    c_{ij}\left(X_i,X_j;\delta, U_{ij}\right) = -\left[A_{i}+B_{j}+X_{i}'\Lambda_{0} X_{j}-U_{ij}\right]
\end{equation}
where $X_{i}$ is a $K \times 1$ vector of mutually exclusive group membership indicators that \textit{is observed} by the econometrician and $\mathbf{U}_i=\left(U_{i1},\ldots,U_{ii-1},U_{ii+1},\ldots,U_{iN}\right)'$ is agent $i$'s vector of idiosyncratic logistic preference shocks over the $N-1$ possible links she can direct (and $\mathbf{U}=\left(\mathbf{U}_1', \ldots , \mathbf{U}_N'\right)'$).\footnote{More generally $X_{i}$ enumerates the support points of a collection of (observed) discrete agent-specific regressors (or a partition of this support into $K$ regions).} All agents observe their own, as well as their peers', preference shock vectors. As is standard in game theory, we use, in a small abuse of notation, $\nu_{i}\left(\sigma_{i},\sigma_{-i};\theta, \mathbf{U}_i\right)$ to denote agent $i$'s \emph{expected} utility under the mixed strategy profile $\sigma=\left(\sigma_{i},\sigma_{-i}\right)$.

The first term in \eqref{eq: payoff_function} captures how agent $i$'s utility varies with the entire structure of the network; this may include benefits from direct, as well as indirect connections. The second term in \eqref{eq: payoff_function} captures the net costs agent $i$ pays in order to maintain those links she chooses to direct.

In theoretical work $g_{i}\left(\mathbf{d}\right)$ is often called the \emph{network benefit} function, while $c_{ij}\left(X_i,X_j;\delta, U_{ij}\right)$ would be associated with the cost of forming edge $ij$ \citep[e.g.,][]{Jackson_NetBook08, Goyal_Book2021}. These costs are generally assumed constant in theory research, while -- as is appropriate given the empirical context -- they are heterogeneous across agents and links here.\footnote{\cite{Johnson_Gilles_RED2000} study the implications of cost heterogeneity on equilibrium network structure in the ``connections" model.} 

While the benefit-cost typology is useful for developing intuitions about the form of NE in this setting, what is essential here is that the first term may vary arbitrarily with $\mathbf{d}$, and hence with peer actions, while the second term is invariant to peers' actions and, furthermore, additively separable in own actions. In some setting the second term in \eqref{eq: payoff_function} may be positive, as occurs when links generate intrinsic surplus. It what follows we call the (negative of the) $j^{th}$ summand in the second part of \eqref{eq: payoff_function} the \emph{baseline utility} that $i$ gets from directing edge $ij$. Of course the appropriate nomenclature is context-specific.

\subsubsection*{Baseline utility}
Considering baseline utility first, we see it is increasing in the heterogeneity terms, assumed unobserved by the econometrician, $A_{i}$ and $B_{j}$. Agents with high values of \emph{out-degree heterogeneity} $A_{i}$ get a large amount of baseline utility from any link they send. In a social network context high $A_{i}$ agents are ``extroverts". Agents with high \emph{in-degree heterogeneity} $B_{j}$, in contrast, are especially attractive targets, or alters, for links sent by others. In a social network high $B_{j}$ agents are ``prestigious" or ``popular".\footnote{Alternatively we can think of high $A_i$ agents as being able to direct links at low cost, and high $B_j$ agents as being low cost alters.}

The $X_{i}'\Lambda_{0} X_{j}\overset{def}{\equiv}W_{ij}'\lambda_{0}$ term allows baseline utility to depend on whether agents assortatively match on their attributes.\footnote{We define $W_{ij}=\left(X_{i}\otimes X_{j}\right)$ and $\lambda=\mathrm{vec}\left(\Lambda'\right)$.} The elements of the $K\times K$ matrix $\Lambda=\left[\lambda_{kl}\right]$ parameterize the systematic utility generated by links, say, from group $k$ to group $l$. For example, in a social network girls might, all things equal, prefer other girls as friends. The $\Lambda_{0}$ matrix parameterizes homophily (or heterophily) of this type.

We leave the joint distribution of $(A_{i},B_{i},X_{i}')'$ unrestricted.\footnote{This distribution does have implications for test power, as will become apparent below. We also comment that $\left\{(A_{i},B_{i},X_{i}')'\right\}_{i=1}^{N}$ need not be i.i.d. There is no requirement, for example, that the agents in the network are a random sample from some population.} This implies that the unobserved degree heterogeneity $(A_{i},B_{i})'$ may be correlated with the observed covariates $X_{i}$, as in fixed effects panel data analyses. 

The final component of baseline utility is idiosyncratic; we assume that the $\{ U_{ij}\} _{i\neq j}$ are independent and identically distributed (iid) logistic random variables. The logistic assumption generates exponential family structure which we exploit when forming our test.

Equation \eqref{eq: payoff_function} with $\gamma=0$ gives agent preferences under our baseline or null model (essentially the dyadic link formation model introduced by \cite{Charbonneau_EJ17}. This model, when fitted by maximum likelihood, can successfully match many features of real world networks. Specifically arbitary in- and out-degree sequences and assortative linking patterns on discrete agent attributes \citep[cf.,][]{Graham_HBE20}. Of course, we are especially interested in settings where the \cite{Charbonneau_EJ17} model does not provide a good description of the network in hand.

\subsubsection*{Network benefit function}
When $\gamma>0$, the first term in \eqref{eq: payoff_function} -- the \emph{network benefit} function $g_{i}(\mathbf{d})$ -- enriches the baseline model to allow agent preferences over links to vary with the presence or absence of links elsewhere in the network. The researcher is free to specify the network benefit function as desired. A few selected examples, drawn from recent theoretical work on strategic network formation, gives a sense of the range of possibilities.

\begin{figure}
\caption{Network benefit function examples}
\begin{centering}
\includegraphics[scale=0.75]{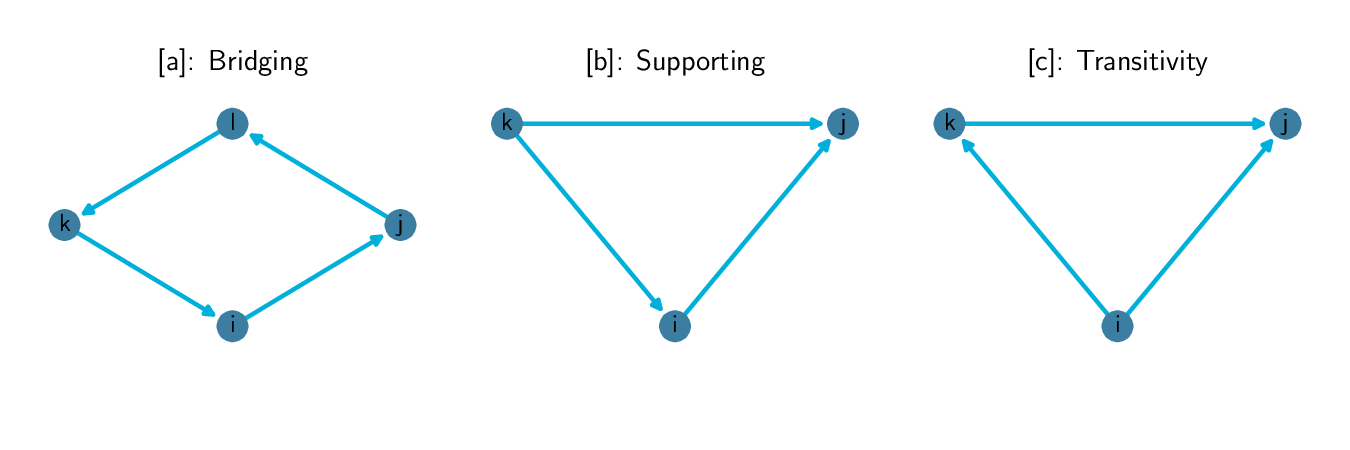}
\par\end{centering}
\caption*{\underline{Source:} Authors' calculations.\hfill\break
\underline{Notes:} Panel [a]: agent $i$ is a bridge from $k$ \emph{to} $j$ and agent $l$ is a bridge from $j$ to $k$. Panel [b]: edge $ij$ is supported by agent $k$. Panel [c]: adding edge $ij$ generates a transitive triad.} 
\end{figure}\label{fig: Network_Benefit_Examples}

\begin{example} \label{ex: connections} \textsc{(Connections)}
    In a seminal paper, \cite{Jackson_Wolinsky_JET96}, introduced the \emph{connections model}. In a directed variant, \cite{Bala_Goyal_EM00} set $g_{i}\left(\mathbf{d}\right)=\sum_{i\neq j}\phi\left(\ell_{ij}\left(\mathbf{\tilde{d}}\right)\right)$ where $\tilde{\mathbf{d}}$ is the undirected network obtained from $\mathbf{d}$  (i.e., $\mathbf{\tilde{d}}=\left[\tilde{d}_{ij}\right]$ with $\tilde{d}_{ij}=1-\left(1-d_{ij}\right)\left(1-d_{ji}\right)$), $\phi:\left\{ 1,2,\ldots,N-1\right\} \rightarrow\mathbb{R}$ is a known function with $\phi\left(k\right)>\phi\left(k+1\right)>0$ for any $k=1,2,\ldots,N-1$, and $\ell_{ij}\left(\mathbf{\tilde{d}}\right)$ the shortest path length between agents $i$ and $j$ in $\mathbf{\tilde{d}}$. Agents prefer to be close to other agents in the network in order to easily access their information, but also wish to maintain as few links as possible, since links are costly to direct. Strong externalities arise in this model: edge $ij$ may incidentally reduce the shortest path length between agents $k$ and $l$, but such benefits are not internalized by agent $i$. Also, since information flows bidirectionally, both agents $i$ and $j$ benefit from edge $ij$, while the cost is shouldered by $i$ alone. 
\end{example}

\begin{example} \label{ex: bridging} \textsc{(Structural Hole / Bridging)}
    \cite{Kleinberg_et_al_ACM2008} introduce a model of network formation inspired by Burt's \citeyearpar{Burt_StructuralHoles1995} theory of “structural holes”. Burt \citeyearpar{Burt_StructuralHoles1995} argued that individuals that connect disparate groups within a network gain “bridging”, “middle-person” or intermediation benefits. Such benefits arise from lying on a (shortest) path connecting two agents not directly connected themselves. Citing empirical evidence, \cite{Kleinberg_et_al_ACM2008} emphasize the special benefits of lying on length two paths between disconnected agents.\footnote{“[T]here appears to be much less measurable benefit to $u$ if it is the internal node on a path between two nodes at graph distance greater than two” \citep[][p. 285]{Kleinberg_et_al_ACM2008}.} If $d_{ki}d_{ij}\left(1-d_{kj}\right)=1$, then $i$ serves as a ``bridge" between $k$ and $j$ (see Panel [a] of Figure \ref{fig: Network_Benefit_Examples}). The summation $\sum_{l}d_{kl}d_{lj}\left(1-d_{kj}\right)$ yields a count of the total number of bridging agents between $k$ and $j$. While agents benefit from serving as a bridge between two agents, these benefits decline in the number of other agents also serving as bridges for the same (directed) dyad. This yields a network payoff function of the form $g_{i}\left(\mathbf{d}\right)=\sum_{j}\sum_{k\neq j}\phi\left(d_{ki}d_{ij}\left(1-d_{kj}\right),\sum_{l}d_{kl}d_{lj}\left(1-d_{kj}\right)\right)$ with $\phi\left(0,k\right)\equiv0$ and $\phi\left(1,k\right)>\phi\left(1,k+1\right)>0$ for $k=1,\ldots,N-2$. See \cite{Goyal_Vega-Redondo_JET2007} for a related model.\footnote{We could, inspired by \cite{Freeman_SM1977}, also consider the model where agents directly value their network betweeness centrality such that $g_{i}\left(\mathbf{d}\right)=\frac{1}{\left(N-1\right)\left(N-2\right)}\sum_{j,k\in\mathcal{N}\setminus\left\{ i\right\} }\frac{\#\text{ of shortest paths from agents \ensuremath{j} to \ensuremath{k} which pass through \ensuremath{i}}}{\#\text{ of shortest paths from agents \ensuremath{j} to \ensuremath{k}}}$.}
\end{example}

\begin{example} \label{ex: support} \textsc{(Supported Links, Transitivity, Reciprocity)}
    \cite{Jackson_et_al_JEL17} introduce a model where agents value supported links. Edge $ij$ is supported by agent $k$ if $d_{ij}d_{ki}d_{kj}=1$ (see Panel [b] of Figure \ref{fig: Network_Benefit_Examples}). This configuration allows agent $k$ to monitor, or referee, relationship $ij$, making it more valuable. This suggests a network benefits function of $g_{i}\left(\mathbf{d}\right)=\sum_{j}d_{ij}\left(\sum_{k}d_{ki}d_{kj}\right)$. If, instead, agents value reciprocity we would set $g_{i}\left(\mathbf{d}\right)=\sum_{j}d_{ij}d_{ji}$; while if they value transitivity in links we would set $g_{i}\left(\mathbf{d}\right)=\sum_{j}d_{ij}\left(\sum_{k}d_{ik}d_{kj}\right)$.
\end{example}

\subsubsection*{Marginal utility}
Let, in an abuse of notation, $\nu_{i}\left(\mathbf{d}\right)\equiv\nu_{i}\left(\mathbf{d}_{i},\mathbf{d}_{-i};\theta, \mathbf{U}_i\right)$; the marginal utility of arc $ij$ for agent $i$ equals
\begin{equation} \label{eq: marginal_utility}
    MU_{ij}\left(\mathbf{d}\right)=\left\{\begin{array}{cc}
    \nu_{i}\left(\mathbf{d}\right)-\nu_{i}\left(\mathbf{d}-ij\right) & \text{if}\thinspace d_{ij}=1\\
    \nu_{i}\left(\mathbf{d}+ij\right)-\nu_{i}\left(\mathbf{d}\right) & \text{if}\thinspace d_{ij}=0
    \end{array}\right.    
\end{equation}
Marginal utility measures the utility gain (loss) to agent $i$ from adding (subtracting) link $ij$ holding the structure of all other links in the network constant (including any other links agent $i$ directs). 
The component of marginal utility associated with the network benefit function $g_{i}\left(\mathbf{d}\right)$ plays an important role in our analysis. Define the marginal network payoff associated with agent $i$ directing a link to $j$ as
\begin{equation} \label{eq: s_ij}
    s_{ij}\left(\mathbf{d}\right)=\left\{\begin{array}{cc}
    g_{i}\left(\mathbf{d}\right)-g_{i}\left(\mathbf{d}-ij\right) & \text{if}\thinspace d_{ij}=1\\
    g_{i}\left(\mathbf{d}+ij\right)-g_{i}\left(\mathbf{d}\right) & \text{if}\thinspace d_{ij}=0
    \end{array}\right.
\end{equation}
Using \eqref{eq: payoff_function} and definition \eqref{eq: s_ij} yields a marginal utility for arc $ij$ of 
\begin{equation}
    MU_{ij}\left(\mathbf{d}\right)=A_{i}+B_{j}+W_{ij}'\lambda_{0}+\gamma_{0}s_{ij}\left(\mathbf{d}\right)-U_{ij}.
\end{equation}
As it features in the computation of the optimal test statistic introduced below, it is helpful to derive the form of $s_{ij}\left(\mathbf{d}\right)$ for the example network benefit functions introduced earlier.

\vspace{5mm}
\noindent \textbf{Example \ref{ex: connections}.} \textsc{(Connections)} In the connections model, when $i$ directs a link to $j$ she weakly reduces her shortest path length to all other agents in the network. In this model $s_{ij}\left(\mathbf{d}\right) \geq 0$ for all $\mathbf{d} \in \mathbb{D}_N$. While there is no closed form expression for $s_{ij}\left(\mathbf{d}\right)$ in the connections model, it is straightforward to compute shortest path lengths between agents numerically (many network manipulation software libraries include routines to do this).  If removing (adding) arc $ij$ increases (decreases) $i$'s distance to many other agents in the network, then $s_{ij}\left(\mathbf{d}\right)$ will be large.

\vspace{5mm}
\noindent \textbf{Example \ref{ex: bridging}.} \textsc{(Structural Hole / Bridging)} 
For the bridging network benefit function $s_{ij}\left(\mathbf{d}\right)$ equals
\begin{equation*}
    s_{ij}\left(\mathbf{d}\right)=\sum_{k\neq j}\phi\left(d_{ki}\left(1-d_{kj}\right),1 + \sum_{l \neq i}d_{kl}d_{lj}\left(1-d_{kj}\right)\right).
\end{equation*}
The marginal utility of edge $ij$ is therefore increasing in the number of agents $k$ which direct edges to $i$, but not to $j$. It is decreasing in the number of agents $l$ and $k$ in which edges $kl$ and $lj$ are present (but edge $kj$ is not). 

\vspace{5mm}
\noindent \textbf{Example \ref{ex: support}.} \textsc{(Supported Links, Transitivity, Reciprocity)} In the support model $s_{ij}\left(\mathbf{d}\right)=\sum_{k}d_{ki}d_{kj}$, which is simply a count of how many agents would support edge $ij$ if it were formed. When agents have a taste for transitivity we have instead
\begin{equation*}
    s_{ij}\left(\mathbf{d}\right)=\sum_{k}d_{ik}d_{kj}+\sum_{k\neq j}d_{ik}d_{jk}
\end{equation*}
which is a count of how many transitive triads (involving agent $i$) would be created if edge $ij$ is added. Finally if agents have a taste for reciprocity we have $s_{ij}\left(\mathbf{d}\right)=d_{ji}$; indicating that the marginal utility of edge $ij$ varies with the presence or absence of the reciprocal edge $ji$.

\subsection{Equilibrium networks}

We assume that the observed network $\mathbf{D}$ coincides with the equilibrium outcome of an $N$-player complete information game. Each agent (i) observes $\left\{(A_{i},B_{i},X_{i}')\right\}_{i=1}^{N}$ and $\left\{U_{ij}\right\}_{i \neq j}$ and then (ii) decides which, out of the $N-1$ other agents, to send links to. Agents may play mixed strategies.

A mixed strategy profile $\sigma^{*}$ is a NE when $\theta = \theta_0$ and $\mathbf{U}=\mathbf{u}$, if for all $i=1,\ldots,N$,
\begin{equation}\label{eq:nash_equilibrium}
\nu_{i}\left(\sigma_{i}^{*},\sigma_{-i}^{*};\theta_0,\mathbf{u}_i\right)\geq\nu_{i}\left(\mathbf{d}_{i},\sigma_{-i}^{*};\theta_0,\mathbf{u}_i\right)
\end{equation}
for all possible pure strategy selections $\mathbf{d}_{i}$. We assume that the \emph{observed} network $\mathbf{D}$ is either a pure strategy NE or in the support of a mixed strategy NE.\footnote{Observe that agent $i$ must consider $2^{N-1}$ different pure strategy deviations in order to verify that their chosen strategy is optimal. This may be unrealistic when $N$ is large. A weaker equilibrium requirement, akin to the notion of pairwise stability introduced by \cite{Jackson_Wolinsky_JET96} for undirected networks, is to require agents to only consider the effects of adding or deleting a \emph{single} link at time on their utility.

Under this weaker stability notion, which we call \emph{single deviation stable} (SDS), we only require that the marginal utility of any link present in the network is non-negative, while that of any link not present is negative. This implies that the observed network $\mathbf{D}$ satisfies the $N\left(N-1\right)$ non-linear equations
\begin{equation*}
    D_{ij}=\mathbf{1}\left(A_{i}+B_{j}+W_{ij}'\lambda_{0}+\gamma_{0}s_{ij}\left(\mathbf{D}\right)\geq U_{ij}\right)
\end{equation*}
for $i,j=1,\ldots,N$ and $j\neq i$.While we maintain the NE assumption in what follows, it turns out that our test is also valid if, instead, the observed network is only SDS. Although single deviation stability is a natural directed analog of pairwise stability, we are not aware of this equilibrium concept being considered before.}

\begin{assumption} \label{ass: DGP} \textsc{(Data generating process)}
Let $\mathbf{U}$ be an $N(N-1)$ vector of iid logistic link preference shocks observed by all agents and $\theta_0 \in \Theta$ be the parameter indexing the payoff function \eqref{eq: payoff_function}. The observed network $\mathbf{D}$ is either a pure strategy NE or contained in the support of a mixed strategy NE of the strategic form game $\left(\mathcal{V}, \mathbb{D}_N, \left\{\nu_{i}\left(\cdot,\cdot;\theta_0,\mathbf{U}_i\right)\right\}_{i \in \mathcal{V}} \right).$
\end{assumption}

\subsection{Likelihood}

In the presence of multiple NE, Assumption \ref{ass: DGP} imposes no restrictions on which one is actually realized in the observed network. Our strategic network formation model is \emph{incomplete}. Although we remain agnostic about equilibrium selection, it is nevertheless useful to develop a notation for, and establish some properties of, the unknown equilibrium selection rule. This allows us to write down a (well-defined) likelihood for the network, albeit abstractly. 

Let $\mathcal{N}(\mathbf{d},\mathbf{u};\theta)$ be a function which assigns, for $\mathbf{U}=\mathbf{u}$, a probability weight to network $\mathbf{d}$: 
\begin{equation}\label{eq:equilibrium_selection}
\mathcal{N}(\mathbf{d}, \mathbf{u};\theta):\mathbb{D}_{N}\times\mathbb{R}^{n}\rightarrow[0,1]    
\end{equation}
In order for $\mathcal{N}(\mathbf{d}, \cdot;\theta)$ to be a valid NE selection function it must satisfy the conditions of Definition \ref{def: equilibrium_selection}.

\begin{definition} \label{def: equilibrium_selection} \textsc{(Equilibrium Selection Function)} For $\mathbf{U}=\mathbf{u}$ the realized vector of logistic link preference shocks and $\theta_0$ the payoff function parameter, let $\mathbf{d}^{*}\left(\mathbf{u};\theta_0 \right)$ be a pure strategy NE or a network contained in the support of a mixed strategy NE and $\mathbb{D}^{*}_N\left(\mathbf{u};\theta_0 \right)$ be the set of all such networks. Function \eqref{eq:equilibrium_selection} is such that (i) $\mathcal{N}\left(\mathbf{d},\mathbf{u};\theta_0\right) \geq 0$ for all $\mathbf{d} \in \mathbb{D}^{*}_N\left(\mathbf{u};\theta_0 \right)$ (ii) $\sum_{\mathbf{d}\in\mathbb{D}^{*}_N\left(\mathbf{u};\theta_0 \right)}\mathcal{N}\left(\mathbf{d},\mathbf{u};\theta_0\right)=1$ and (iii) $\mathcal{N}\left(\mathbf{d},\mathbf{u};\theta_0\right) = 0$ for all $\mathbf{d} \in  \mathbb{D}_N \backslash \mathbb{D}^{*}_N\left(\mathbf{u};\theta_0 \right)$.
\end{definition}

If $\mathcal{N}(\mathbf{d},\cdot;\theta)$ satisfies the conditions of Definition \ref{def: equilibrium_selection}, then the likelihood of observing network $\mathbf{D}=\mathbf{d}$ is
\begin{equation}\label{eq:likelihood}
P\left(\mathbf{d};\theta,\mathcal{N}\right)=\int_{\mathbf{u}\in\mathbb{R}^{n}}\mathcal{N}(\mathbf{d},\mathbf{u};\theta)f_{\mathbf{u}}(\mathbf{u})\mathrm{d}\mathbf{u},
\end{equation}
where $f_{\mathbf{u}}(\mathbf{u})=\prod_{i\neq j}f_{U}(u_{ij})$ with $f_{U}(u)=e^{u}/[1+e^{u}]^{2}.$ Of course, for the likelihood \eqref{eq:likelihood} to be well-defined we require that $\mathcal{N}(\mathbf{d}, \cdot;\theta)$ is measureable.

\begin{theorem} \label{thm: measurability} \textsc{(Likelihood)} For any network $\mathbf{d}\in\mathbb{D}_{N}$ there exists a measurable function $\ensuremath{\mathcal{N}(\mathbf{d},\cdot;\theta)}\thinspace:\thinspace\mathbb{R}^{n}\rightarrow\left[0,1\right]$, which assigns to $\mathbf{u}\in\mathbb{R}^{n}$ a $\mathrm{NE}$ weight on the pure strategy combination corresponding to $\mathbf{d}$.
\end{theorem}

The proof of Theorem \ref{thm: measurability} can be found in Appendix \ref{measurablity}. 

\section{Testing for strategic interaction}\label{sec:test}

The development in this section parallels the first and second use cases outlined in the introduction. We first discuss how to assess the adequacy of the baseline model as a description of the network in hand. Utilizing a conditioning argument we construct an \emph{exact} test (up to simulation error) of the null of ``correct specification". 
An alternative model is not explicitly formulated in this case, although researcher intuitions about plausible directions of mis-specification typically guides the choice of test statistic. As shown by \cite{Lehmann_Romano_TSH05}, it is impossible to construct a test with power in all possible directions of mis-specification.

Next we consider applications where the analyst carefully specifies the alternative model (through an explicit choice of the network benefit function, $g_i\left(\mathbf{d}\right)$). Here the researcher believes the true network formation model lies in either the null or the (specified) alternative model space; the purpose of testing is to determine which situation prevails. In this second application we seek to construct a test which rejects with high probability when the alternative is true, while continuing to control size under the null. 

Throughout, and crucially, we wish to remain agnostic about the distribution of any degree heterogeneity across agents as well as the form of any homophily and/or heterophily. Let $\Delta$ denote a subset of the $K^{2}+2N$ dimensional Euclidean space in which $\delta_{0}=\left(\lambda_{0},\mathbf{A}_{0},\mathbf{B}_{0}\right)$ is, a priori, known to lie, and 
\begin{equation}\label{eq:null_parameter_space}
\Theta_{0}=\left\{ (\gamma,\delta')\thinspace:\thinspace\gamma=0,\delta\in\Delta\right\}.    
\end{equation}
Our null hypothesis is the \textit{composite} one: 
\begin{equation}\label{eq:null_hypothesis}
H_{0}\thinspace:\thinspace\theta\in\Theta_{0}   
\end{equation}
since $\delta$ may range freely over $\Delta\subset\mathbb{R}^{K^{2}+2N}$ under the null.

Under the null the likelihood is $P_{0}(\mathbf{d};\delta)\overset{def}{\equiv}P(\mathbf{d};(0,\delta')',\mathcal{N}_{0})$ with
\begin{align*}
\mathcal{N}_{0}(\mathbf{d},\mathbf{u};\theta)=&\prod_{i}\prod_{j}\mathbf{1}\left(A_{i}+B_{j}+W_{ij}'\lambda\geq u_{ij}\right)^{d_{ij}} \\
 	                                          &\times\mathbf{1}\left(A_{i}+B_{j}+W_{ij}'\lambda<u_{ij}\right)^{1-d_{ij}}.
\end{align*}
Under the null the unique ``equilibrium'' network is the one where all links with positive marginal utility are present and those with negative marginal utility are not. These marginal utilities are invariant to the presence or absence of links elsewhere in the network; $\mathcal{N}_{0}(\mathbf{d},\mathbf{u};\theta)$ places a probability of $1$ on this network. Evaluating the integral \eqref{eq:likelihood} under the null yields
\begin{align*}
P_{0}(\mathbf{d};\delta)	=	&\prod_{i=1}^{N}\prod_{j\neq i}\left[\frac{\exp\left(W_{ij}'\lambda+R_{i}'\mathbf{A}+R_{j}'\mathbf{B}\right)}{1+\exp\left(W_{ij}'\lambda+R_{i}'                                  \mathbf{A}+R_{j}'\mathbf{B}\right)}\right]^{d_{ij}} \\		                       &\times\left[\frac{1}{1+\exp\left(W_{ij}'\lambda+R_{i}'\mathbf{A}+R_{j}'\mathbf{B}\right)}\right]^{1-d_{ij}}
\end{align*}
where $R_i$ is the $N\times1$ vector with a $1$ in its $i^{th}$ element and zeros elsewhere.\footnote{Variants of this likelihood are analyzed by \cite{Chatterjee_et_al_AAP11}, \cite{Charbonneau_EJ17}, \cite{Graham_EM17}, \cite{Jochmans_JBES18}, \cite{Dzemski_RESTAT18} and \cite{Yan_et_al_JASA18}.} 

\subsection{Use case 1: exact goodness-of-fit testing}
Under the null our likelihood, $P_{0}(\mathbf{d};\delta)$, is a member of the exponential family. To see this it is helpful to establish some additional notation. The out- and in-degree sequences equal:
\begin{equation}
\mathbf{S}         =\left(\begin{array}{c} \mathbf{S}_{\mathrm{out}}\\ 
                                           \mathbf{S}_{\mathrm{int}}
\end{array}\right)'
                   =\left(\begin{array}{c} D_{1+},\ldots,D_{N+}\\
                                           D_{+1},\ldots,D_{+N}
\end{array}\right).    
\end{equation}
Here $D_{+i}=\sum_{j}D_{ji}$ and $D_{i+}=\sum_{j}D_{ij}$ equal the in- and out-degree of agents $i=1,\ldots,N$.

The $K\times K$ \textit{cross-link matrix} equals 
\begin{equation}
\mathbf{M}=\sum_{i}\sum_{j}D_{ij}X_{i}X_{j}'.    
\end{equation}
This matrix summarizes the inter-group link structure in the network (homophily). The $kl^{th}$ element of $\mathbf{M}$ records the number of links sent by type $k$ agents (e.g., semiconductor manufacturers) to type $l$ agents (e.g., computer manufacturers).

Let $\mathbf{S},\mathbf{M}$ be a degree sequence and cross-link matrix. We say $\mathbf{S},\mathbf{M}$ is \textit{graphical} if there exists at least one arc set $\mathcal{A}$ such that $G\left(\mathcal{V},\mathcal{A}\right)$ is a simple directed graph with degree sequence $\mathbf{S}$ and cross link matrix $\mathbf{M}$. We call any such network a \emph{realization} of $\mathbf{S},\mathbf{M}$. The set of all possible realizations of $\mathbf{S},\mathbf{M}$ is denoted by $\mathbb{G}_{\mathbf{S},\mathbf{M}}$ ($\mathbb{D}_{\mathbf{S},\mathbf{M}}$ denotes the associated set of adjacency matrices).

With this notation it is easy to verify that the null model belongs to the exponential family (see \cite{Graham_EM17}):
\begin{equation}
P_{0}(\mathbf{d};\delta)=c(\delta)\exp\left(\mathbf{t}'\delta\right),\thinspace:\thinspace\delta\in\Delta\    
\end{equation}
with a (minimally) sufficient statistic for $\delta$ of $\mathbf{t}=\left(\mathrm{vec}\left(\mathbf{m}'\right)',\mathbf{s}_{\mathrm{out}}',\mathbf{s}_{\mathrm{in}}'\right)'$. In words, the $K^{2}+N+N$ sufficient statistics are (i) the cross link matrix, (ii) the out-degree sequence and (iii) the in-degree sequence.

Under $H_{0}$ the conditional likelihood of the event $\mathbf{D}=\mathbf{d}$ is 
\begin{equation} \label{eq:conditional_likelihood}
P_{0}\left(\left.\mathbf{d}\right|\mathbf{T=t}\right)=\frac{P_{0}\left(\mathbf{d};\delta\right)}{\sum_{\mathbf{v}\in\mathbb{D}_{\mathbf{s},\mathbf{m}}}P_{0}\left(\mathbf{v};\delta\right)}=\frac{1}{\left|\mathbb{D}_{\mathbf{s},\mathbf{m}}\right|}
\end{equation}
if $\mathbf{d} \in \mathbb{D}_{\mathbf{s},\mathbf{m}}$ and zero otherwise. Under the null of no strategic interaction all networks with the same in- and out-degree sequences and cross link structure are equally likely. Importantly this conditional likelihood is invariant to the actual value of the nuisance parameter $\delta$. 

By conditioning on $\mathbf{T}$, which is sufficient for $\delta$, we isolate the information in the data that is relevant for assessing model adequacy \citep{Barndorff-Nielsen_Cox_IA1994}. This follows because conditional on $\mathbf{T}$, the null model \emph{completely} specifies the distribution of $\mathbf{D}$. Consequently, the distribution of any statistic of the adjacency matrix, say $R\left(\mathbf{D}\right)$, is also fully specified. Specifically the null distribution $R\left(\mathbf{D}\right)$ is the one induced by a discrete uniform distribution on $\mathbb{D}_{\mathbf{S},\mathbf{M}}$: 
\begin{equation}
    \Pr\left(\left.R\left(\mathbf{D}\right)\leq r\right|\mathbf{T};\theta\in\Theta_{0}\right)=\frac{1}{\left|\mathbb{D}_{\mathbf{S,M}}\right|}\sum_{\mathbf{d}\in\mathbb{D}_{\mathbf{S,M}}}\mathbf{1}\left(R\left(\mathbf{d}\right)\leq r\right).
\end{equation}
To test model goodness-of-fit, we simply check whether the value of $R\left(\mathbf{D}\right)$ in the network in hand is at an extreme quantile of this distribution. If it is, we take this as evidence against the baseline (null) model. 

\subsubsection*{Similarity and conditioning}
A test with critical function $\phi\left(\mathbf{D}\right)$ will have size $\alpha$ if its null rejection probability (NRP) is less than or equal to $\alpha$ for \textit{all} values of the nuisance parameter:
\begin{equation}
\underset{\theta\in\Theta_{0}}{\sup}\,\mathbb{E}_{\theta}\left[\phi\left(\mathbf{D}\right)\right]=\underset{\gamma=\gamma_{0},\delta\in\triangle}{\sup}\,\mathbb{E}_{\theta}\left[\phi\left(\mathbf{D}\right)\right]=\alpha.    
\end{equation}
Since the nuisance parameter $\delta$ is very high dimensional, size control is \emph{a priori} non-trivial. For some intuition as to why consider, as an example, the case where $s_{ij}(\mathbf{d})=\sum_{k}d_{ki}d_{kj}$, such that agents' have a taste for supported links when $\gamma_{0}>0$. A natural test statistic in this case would be some function of $\mathbf{D}$ that is increasing in the number of supported links in the network.\footnote{\cite{Jackson_et_al_AER12} suggest the fraction of links in the network which are supported.} The researcher would then reject the null of $\gamma_{0}=0$ when this statistic is large enough. Unfortunately, the expected number of supported links varies dramatically under the null depending on the value of $\delta$. Certain configurations of $\mathbf{A}$, $\mathbf{B}$ and/or $\lambda$ may result in a network with substantial link clustering (and hence support) even when agents' have no taste for support per se. If we choose a single critical value for rejection then, depending on the values of $\mathbf{A}$, $\mathbf{B}$ and/or $\lambda$, size may be very poor.

To avoid any size distortion induced by variation in $\delta$ over $\Delta\subset\mathbb{R}^{K^{2}+2N}$ we exploit the exponential family structure of our model (under the null). Let $\mathbb{T}=\left\{ \left(\mathbf{s},\mathbf{m}\right)\thinspace:\thinspace\mathbf{s},\mathbf{m}\text{ is graphical}\right\}$  be the set of possible sufficient statistics $\mathbf{T}$. Instead of choosing a single critical value, which may result in under- or over-rejection, depending on the value of $\delta$, we proceed conditionally on $\mathbf{T} \in \mathbb{T}$, varying our critical value with $\mathbf{T}$. In this way we ensure good size control.

Formally, for each $\mathbf{t}\in\mathbb{T}$ we form a test with the property that, for all $\theta\in\Theta_{0}$,
\begin{equation}\label{eq:conditional_similarity}
\mathbb{E}_{\theta}\left[\left.\phi\left(\mathbf{D}\right)\right|\mathbf{T}=\mathbf{t}\right]=\alpha.    
\end{equation}
Such an approach ensures \textit{similarity} of our test since, by iterated expectations,
\begin{equation}\label{eq:similarity}
\mathbb{E}_{\theta}\left[\phi\left(\mathbf{D}\right)\right]=\mathbb{E}_{\theta}\left[\mathbb{E}_{\theta}\left[\left.\phi\left(\mathbf{D}\right)\right|\mathbf{T}\right]\right]=\alpha    
\end{equation}
for any $\theta\in\Theta_{0}$ \citep{Ferguson_MS67}. By proceeding conditionally we ensure that the NRP is unaffected by the value of $\delta$.

For any $\mathbf{t}\in\mathbb{T}$ we can construct an \emph{exact} test, as is required by \eqref{eq:conditional_similarity}, because our model completely specifies the distribution of networks conditional on $\mathbf{T}=\mathbf{t}$ under the null. Condition \eqref{eq:similarity} follows immediately. Using some well-known results from the theory of exponential families, we can make the stronger claim that similarity is only possible by conditioning. 
\begin{lemma} \label{lemma: similarity}
\textsc{(Similarity)} Any similar test of $ H_{0}\thinspace:\thinspace\theta\in\Theta_{0}$ conditions on the realized value of $\mathbf{T}$.  
\end{lemma}
\begin{proof}
By \citet[][Lemma 1, Section 3.6]{Ferguson_MS67} $\mathbf{T}$ is a boundedly complete sufficient statistic for $\theta$ under the null. The claim then follows from \citet[][Theorem 2, Section 5.4]{Ferguson_MS67}. 
\end{proof}

\subsubsection*{Implementation}
To operationalize, let $R(\mathbf{D})$ be some statistic of the adjacency matrix. For example $R(\mathbf{D})$ might be the network reciprocity index \citep{Newman_NetBook10}:
\begin{equation}
R(\mathbf{D})=\frac{2\hat{P}_{11}}{2\hat{P}_{11}+\hat{P}_{01}},    
\end{equation}
where
\begin{equation}
\hat{P}_{01}=\frac{2}{N\left(N-1\right)}\sum_{i=1}^{N-1}\sum_{j=i+1}^{N}\left[D_{ij}\left(1-D_{ji}\right)+\left(1-D_{ij}\right)D_{ji}\right]
\end{equation}
equals the fraction of dyads which take an unreciprocated or ``asymmetric" configuration and
\begin{equation}
\hat{P}_{11}=\frac{2}{N\left(N-1\right)}\sum_{i=1}^{N-1}\sum_{j=i+1}^{N}D_{ij}D_{ji}
\end{equation}
the fraction which take a reciprocated or ``mutual" configuration. 

A conditional test based upon $R(\mathbf{D})$ will have a critical function of
\begin{equation} \label{eq: critical_function}
\phi\left(\mathbf{d}\right)=\left\{ \begin{array}{cc}
1 & R\left(\mathbf{d}\right)>c_{\alpha}\left(\mathbf{t}\right)\\
g_{\alpha}\left(\mathbf{t}\right) & R\left(\mathbf{d}\right)=c_{\alpha}\left(\mathbf{t}\right)\\
0 & R\left(\mathbf{d}\right)<c_{\alpha}\left(\mathbf{t}\right)
\end{array}\right.    
\end{equation}
where the values of $c_{\alpha}\left(\mathbf{t}\right)$ and $g_{\alpha}\left(\mathbf{t}\right)\in\left[0,1\right]$ are chosen to satisfy the requirement that $\mathbb{E}_{\theta}\left[\left.\phi\left(\mathbf{D}\right)\right|\mathbf{T}=\mathbf{t}\right]=\alpha$.

Under the null all adjacency matrices with the $\mathbf{S}=\mathbf{s}$ and $\mathbf{M}=\mathbf{m}$ are equally probable. By enumerating all adjacency matrices in $\mathbb{D}_{\mathbf{s},\mathbf{m}}$ we could exactly compute the null distribution of $R\left(\mathbf{D}\right)$ and hence the critical values $c_{\alpha}\left(\mathbf{t}\right)$ and $g_{\alpha}\left(\mathbf{t}\right)$. In general such a brute force approach will be infeasible.\footnote{In fact very little is known about the set $\mathbb{D}_{\mathbf{s},\mathbf{m}}$; for example we are aware of no method for checking whether a given $\mathbf{s},\mathbf{m}$ pair is graphic. From related settings we believe that the cardinality of $\mathbb{D}_{\mathbf{s},\mathbf{m}}$ will typically be intractably huge even for modestly-sized networks. See \cite{Blitzstein_Diaconis_IM11} for discussion of this point and examples from a related setting.} Therefore a method of approximating the exact null distribution is required. The simulation algorithm introduced below provides such a method.

The intuition behind this test is straightforward. If the network in hand has an ``unusually" large value of $R(\mathbf{D})$ relative to the set of all networks with same in- and out-degree sequences and cross-link matrices, then we reject the null that the baseline model is correctly specified. A rejection is not interpreted as evidence in favor of a particular alternative model. Relatedly, a feature of goodness-of-fit tests, including this one, is that we have may have low, or even, power equal to size in certain directions \citep{Lehmann_Romano_TSH05}.

\subsection{Optimal testing with an explicit alternative}

In this section we discuss how to test when the alternative model space is explicitly specified. That is, when the researcher explicitly specifies the network benefit function in \eqref{eq: payoff_function} and proceeds under the premise that the true network generating process lies either in the null or the (explicitly specified) alternative model space. In such settings a rejection provides evidence that $\gamma_0>0$ (in the context of a specific network benefit function). Naturally the researcher would like to maximize her power to reject, while continuing to maintain similarity. To accomplish this requires choosing the right test statistic.

Because equilibrium selection is not specified under the alternative, likelihood ratio (LR) testing is not feasible \citep[cf.,][]{Chen_et_al_EM2018}. As an alternative to a LR test, we instead choose, for each $\mathbf{t}\in\mathbb{T}$, the critical function, $\phi\left(\mathbf{D}\right)$ to maximize the \textit{derivative} of the (conditional) power function $\beta\left(\gamma,\mathbf{t}\right)=\mathbb{E}\left[\left.\phi\left(\mathbf{D}\right)\right|\mathbf{T}=\mathbf{t}\right]$ evaluated at $\gamma=0$ subject to the (conditional) size constraint $\mathbb{E}_{\theta}\left[\left.\phi\left(\mathbf{D}\right)\right|\mathbf{T}=\mathbf{t}\right]=\alpha$. Such a $\phi\left(\mathbf{D}\right)$ is \textit{locally best} \cite[][Lemma 1, Section 5.5]{Ferguson_MS67}. Remarkably we show that the locally best test does not depend upon the form of the equilibrium selection mechanism $\mathcal{N}(\mathbf{d}, \mathbf{u};\theta)$.

Differentiating the power function we get 
\begin{equation}\label{eq:slope_of_power_function}
\left.\frac{\partial\beta\left(\gamma,\mathbf{t}\right)}{\partial\gamma}\right|_{\gamma=0}	=\mathbb{E}\left[\left.\phi\left(\mathbf{D}\right)\mathbb{S}_{\gamma}\left(\left.\mathbf{D}\right|\mathbf{T};\theta\right)\right|\mathbf{T}=\mathbf{t}\right]    
\end{equation}
with $\mathbb{S}_{\gamma}\left(\left.\mathbf{d}\right|\mathbf{t};\theta\right)$ denoting the conditional score function
\begin{align*} 
\mathbb{S}_{\gamma}\left(\left.\mathbf{d}\right|\mathbf{t};\theta\right) &= \frac{1}{P_{0}\left(\mathbf{d};\delta\right)}\left.\frac{\partial           P\left(\mathbf{d};\theta\right)}{\partial\gamma}\right|_{\gamma=0}-\sum_{\mathbf{v}\in\mathbb{D}_{\mathbf{s},\mathbf{m}}}\left.\frac{\partial P\left(\mathbf{v};\theta\right)}{\partial\gamma}\right|_{\gamma=0} \\ 
&=\frac{1}{P_{0}\left(\mathbf{d};\delta\right)}\left.\frac{\partial P\left(\mathbf{d};\theta\right)}{\partial\gamma}\right|_{\gamma=0}+k\left(\mathbf{t}\right)
\end{align*}
and $k\left(\mathbf{t}\right)$ only depending on the data through $\mathbf{T}=\mathbf{t}$ (Here, and in the balance of this section, it is understood that $\delta$ is evaluated at is population value $\delta_0$). By the Neyman-Pearson lemma, the test with the critical function given by equation \eqref{eq: critical_function} above, where the test statistic, $R\left(\mathbf{d}\right)$, is set equal to the log-likelihood gradient, $\frac{1}{P_{0}\left(\mathbf{d};\delta\right)}\left.\frac{\partial P\left(\mathbf{d};\theta\right)}{\partial\gamma}\right|_{\gamma=0}$, will be locally best within the class of similar tests.

The idea behind the locally best test is as follows. If the likelihood increases sharply as we move away from the null \emph{in the direction of the alternative of interest}, then we take this as evidence against the null. Intuitively if the likelihood gradient in the neighborhood of the null is large, then the likelihood ratio will also be large for simple alternatives close to the null (i.e., when $\gamma \in \left(0,\epsilon\right]$). 

Constructing the locally best critical function requires calculating $\frac{1}{P_{0}\left(\mathbf{d};\delta\right)}\left.\frac{\partial P\left(\mathbf{d};\theta\right)}{\partial\gamma}\right|_{\gamma=0}$. This is not straightforward since it depends on properties of the likelihood under the alternative (and consequently the equilibrium selection function). Nevertheless, we are able to derive the form of this derivative.
\begin{theorem}\label{thm: derivative}
\textsc{(Locally Best Test)} (i) $P\left(\mathbf{d};\theta,\mathcal{N}\right)$ is twice differentiable with respect to $\gamma$ at $\gamma=0$.  Its first derivative at $\gamma=0$ is
\begin{multline} \label{eq: locally_best_statistic_1}
\left.\frac{\partial P\left(\mathbf{d};\theta,\mathcal{N}\right)}{\partial\gamma}\right|_{\gamma=0} =	P_{0}\left(\mathbf{d};\delta\right) \\
	\times\left[\sum_{i\neq j}s_{ij}\left(\mathbf{d}\right)\left\{ d_{ij}\frac{f_{U}\left(\mu_{ij}\right)}{\int_{-\infty}^{v_{ij}}f_{U}\left(u\right)\mathrm{d}u}-\left(1-d_{ij}\right)\frac{f_{U}\left(\mu_{ij}\right)}{\int_{v_{ij}}^{\infty}f_{U}\left(u\right)\mathrm{d}u}\right\} \right],    
\end{multline}
recalling that $\mu_{ij}=A_{i}+B_{j}+X_{j}'\Lambda_{0}X_{i}$ equals the systematic, non-strategic, component of utility generated by arc $ij$ and that $f_{U}$ is the logistic density; (ii) the test statistic $R\left(\mathbf{d}\right)=\frac{1}{P_{0}\left(\mathbf{d};\delta\right)}\left.\frac{\partial P\left(\mathbf{d};\theta\right)}{\partial\gamma}\right|_{\gamma=0}$ yields the locally best test in the direction of the specified alternative within the class of similar tests.
\end{theorem}

The proof of Theorem \ref{thm: derivative}, along with some additional commentary, can be found in Section \ref{app: derivative_proof} of the Supplemental Web Appendix. A key implication of Theorem \ref{thm: derivative} is that the form of the locally best test statistics does \emph{not} depend upon $\mathcal{N}$, the equilibrium selection mechanism. This is essential, since optimal testing would not be feasible otherwise (at least without additional assumptions). One intuition for this finding is that equilibrium is unique with high probability when $\gamma$ is close to zero. This means we can effectively ignore draws of $\mathbf{U}$ which lead to multiple equilibria when differentiating the likelihood.

Indeed, when $\gamma$ is close to zero most players will have a strictly dominant strategy (that is the optimal set of links for them to send will be invariant to the play of their peers). Of course we need more information to recover the gradient with respect to $\gamma$, since this parameter measures the responsiveness of agents to their peers' actions. It turns out that a key scenario used in the derivative calculation involves considering draws of $\mathbf{U}$ where all players \emph{except one} have strictly dominant strategies. The one player without a strictly dominant strategy provides the needed gradient information.

\subsubsection*{Locally best vs. heuristic test statistics}
With a little manipulation we can simplify \eqref{eq: locally_best_statistic_1} to:
\begin{equation} \label{eq: locally_best_statistic_2}
\frac{1}{P_{0}\left(\mathbf{d};\delta\right)}\left.\frac{\partial P\left(\mathbf{d};\theta\right)}{\partial\gamma}\right|_{\gamma=0}=\sum_{i\neq j}\left[d_{ij}-F_{U}\left(\mu_{ij}\right)\right]s_{ij}\left(\mathbf{d}\right)    
\end{equation}
where $F_{U}\left(u\right)=e^{u}/\left[1+e^{u}\right]$ is the logistic CDF. This form of the statistic provides insight into how our test accumulates evidence against the null in practice. Consider the case where $s_{ij}\left(\mathbf{d}\right)=d_{ji}$, as would be true in agents' have a taste for reciprocated links. Observe that $F_{U}\left(\mu_{ij}\right)$ corresponds to the probability of the edge  $ij$ under the null. Therefore the optimal test statistic is large if we observe that many $ij$ links \emph{with low probability under the null} are reciprocated. It is not many reciprocated links that drives rejection per se, but the presence of many ``unexpected" reciprocated links. 

Consider a network of boys and girls with agents exhibiting a strong taste for gender-based homophily. The optimal test statistic in this case is the \emph{conditional} sample covariance of $D_{ij}$ and $D_{ji}$ given $(A_i, B_i, X_i)$ and $(A_j, B_j, X_j)$. The test based upon the reciprocity index is -- essentially -- based upon the \emph{unconditional} covariance. The effect of conditioning is to, for example, given more weight to heterophilous reciprocated links than to homophilous ones. Similarly we give more weight to reciprocated links across low degree agents, than to those across high degree agents. 

\subsubsection*{Implementation}

Two practical issues remain. The first, how to simulate the null distribution of the optimal test statistic, is covered in the next section. Second, although the locally best test statistic does not depend on the details of equilibrium selection, it does depend on $\delta_0$. Although the test will remain admissible when $\delta_0$ is replaced by some other, perhaps arbitrary, $\delta$, it will not be locally best.

A practical solution to this problem is to replace $\delta_0$ with its maximum likelihood estimate (MLE) computed under the null. This particular MLE is studied by \cite{Yan_et_al_JASA18}. In our Monte Carlo experiments, some of which are reported in the Supplemental Web Appendix, we have found that replacing $\delta_0$ with its MLE, results in a test which is nearly as powerful as the infeasibe oracle test based on $\delta_0$, and far more powerful that tests based on ad hoc statistics.

\section{Simulation}\label{sec:simulation}

Because a complete enumeration of $\mathbb{D}_{\mathbf{s},\mathbf{m}}$ is not feasible unless $N$ is very small, making our test practical requires a method of constructing uniform random draws from this set. Such draws can be used to simulate the null distribution of any test statistic of interest.

The problem of simulating networks with fixed degree sequences
is well-studied; with many domain specific applications \citep[e.g.,][]{Sinclair_Book1993}. We add to this problem the additional requirement that the simulated network satisfies the cross-link matrix constraint. 

Prior work on network simulation adopts one of two basic approaches. The first approach begins with an empty graph and randomly adds links. Links need to be added such that the end graph satisfies the degree sequence constraint. \cite{Blitzstein_Diaconis_IM11} develop an algorithm along these lines. They cleverly use checks for graphicality of a degree sequence, available in the discrete math literature, to add links in a way which constrains the end graph to be in the target set.\footnote{See also \cite{Genio2010} and \cite{Kim2012}. \cite{Graham_Pelican_BookCh2020} provide a textbook discussion of the \cite{Blitzstein_Diaconis_IM11} algorithm.} 

The second approach, to which our new method belongs, uses Markov Chain Monte Carlo (MCMC). Specifically an initial graph, satisfying the target constraints, is randomly rewired many times to create a new graph from the target set. Key to this approach is ensuring that each rewiring is compatible with the target constraints (e.g., maintains the network's degree sequence). The algorithm also needs to be constructed carefully to ensure that the end graph is a \emph{uniform} random draw from the target set. \cite{Sinclair_Book1993}, \cite{Rao_et_al_Sankhya96}, \cite{McDonald_et_al_SN07}, \cite{Berger2010} and \cite{Tao_NS16} all developed MCMC methods for simulating graphs (or digraphs) with given degree sequences.

We are aware of no method of generating adjacency matrix draws from $\mathbb{D}_{\mathbf{s},\mathbf{m}}$. The novelty of this problem, relative to the work described above, is the presence of the additional cross link matrix constraint, $\bf{M}$. In the discrete math literature the cross link matrix constraint corresponds to what is called a partition adjacency matrix (PAM) constraint. \cite{czabarka2017algebraic} conjecture that determining whether a given $\mathbf{s}, \mathbf{m}$ pair is graphical, the PAM realization problem, is NP-complete. If their conjecture is correct (and NP $\neq$ P), using a \cite{Blitzstein_Diaconis_IM11} type algorithm to draw from $\mathbb{D}_{\mathbf{s},\mathbf{m}}$ is not feasible. 

This leaves MCMC methods. \cite{Erdoes2017} showed that naively incorporating a PAM constraint into existing MCMC algorithms destroys their correctness. In this section we introduce a new MCMC algorithm that \emph{does} generate uniform random draws from $\mathbb{D}_{\mathbf{s},\mathbf{m}}$. This algorithm is of independent interest. Before describing the algorithm we introduce some additional definitions and notation.

\subsection{Notation and definitions}
We start by defining an alternating walk.

\theoremstyle{definition}

\begin{samepage}
\begin{definition}{\textsc{(Alternating Walk)}}\label{def:alternating_walk}
An alternating walk $H$ is sequence of (ordered) dyads of the form
\begin{equation}
    H:=\left(i_{1},i_{2}\right),\left(i_{3},i_{2}\right),\left(i_{3},i_{4}\right),\ldots,\left(i_{l},i_{l-1}\right)
\end{equation}
or 
\begin{equation}
    H:=\left(i_{2},i_{1}\right),\left(i_{2},i_{3}\right),\left(i_{4},i_{3}\right),\ldots,\left(i_{l-1},i_{l}\right)
\end{equation}
with $i_{k}\in\mathcal{V}\left(G\right)$,  $i_{k}\neq i_{k+1}$,  $i_{k}\neq i_{k-1}$and\\
(i) if  $\left(i_{k},i_{k-1}\right)\in\mathcal{A}\left(G\right)$, then $ \left(i_{k},i_{k+1}\right)\notin\mathcal{A}\left(G\right)$\\
(ii) if $\left(i_{k},i_{k-1}\right)\notin \mathcal{A}\left(G\right)$, then $ \left(i_{k},i_{k+1}\right)\in \mathcal{A}\left(G\right)$\\
(ii) if $\left(i_{k-1},i_{k}\right)\in\mathcal{A}\left(G\right)$, then $ \left(i_{k+1},i_{k}\right)\notin\mathcal{A}\left(G\right)$\\
(iv) if $\left(i_{k-1},i_{k}\right)\not \in \mathcal{A}\left(G\right)$, then $ \left(i_{k+1},i_{k}\right)\in\mathcal{A}\left(G\right)$\\
for all $k=2,\ldots,l-1$.
\end{definition}
\end{samepage}

For brevity we will often refer to a walk simply by its node sequence, writing $H:=i_{1}i_{2},\ldots,i_{l}$. To unpack Definition \ref{def:alternating_walk} it is easiest to consider an example. In Figure \ref{fig: feasible_schlaufen_sequence}, Panel B, three altering walks are shown (the links not present are depicted as dotted arrows).

Observe that for $H:=i_{1}i_{2},\ldots,i_{l}$, the adjacency matrix entries $D_{i_{1}i_{2}},D_{i_{3}i_{2}},\ldots,D_{i_{l}i_{l-1}}$ alternate between ones and zeros (or zeros and ones). This observation suggests a method of constructing an alternating walk via a sequence of ``hops" across the adjacency matrix: pick row $i_{1}$ of the adjacency matrix and move horizontally to column $i_{2}$, where $i_{2}$ corresponds to one of the agents to which $i_{1}$ directs a link, next move vertically to row $i_{3}$, where $i_{3}$ is an agent which does not direct a link to $i_{2}$, and so on.\footnote{This description is essentially due to \cite[][p. 124]{Tao_NS16}.} We call the horizontal moves \textit{active steps} and vertical moves \textit{passive steps}. Figure \ref{fig: constructing_alternating_walks} provides an example construction. The different cases in Definition \ref{def:alternating_walk} correspond to walks beginning/ending with passive/active steps.

\begin{figure}
\caption{\label{fig: constructing_alternating_walks}Constructing an alternating walk}
\begin{centering}
\includegraphics[scale=0.60]{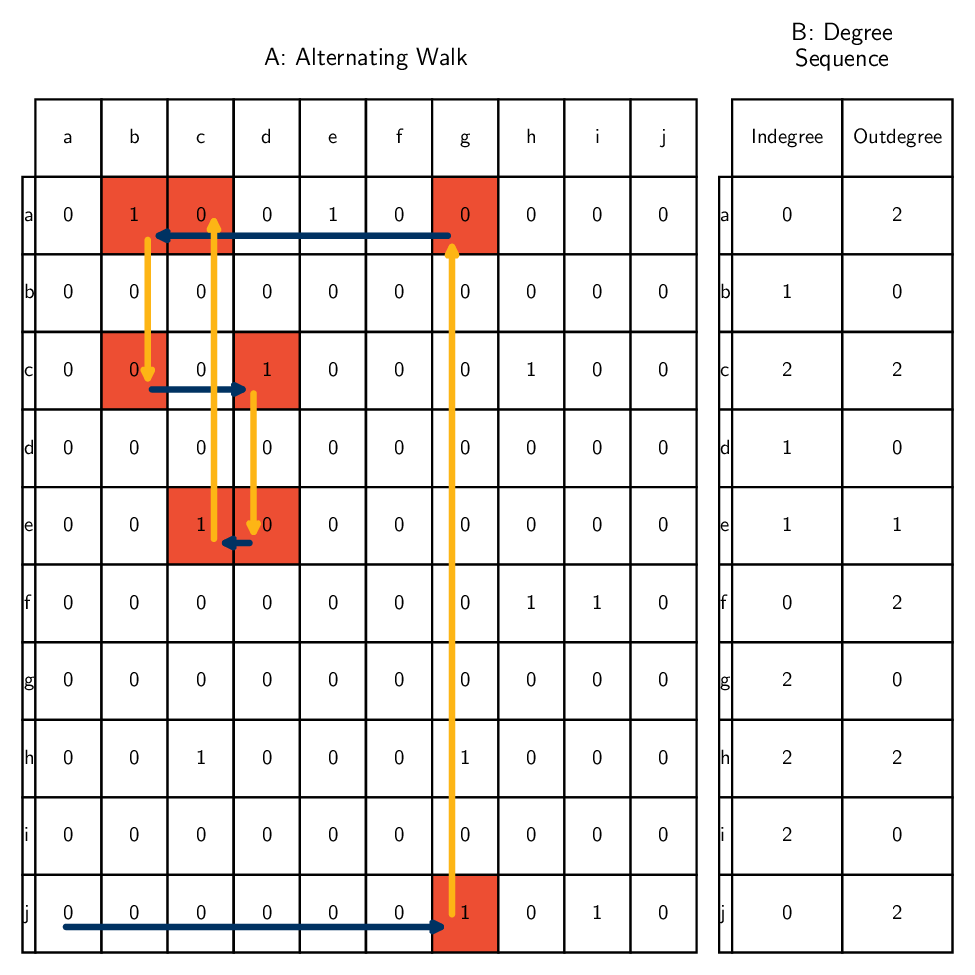}
\par\end{centering}
\caption*{
\underline{Source:} Authors' calculations. \hfill\break
\underline{Notes:} Panel A depicts an alternating walk $j,g,a,b,c,d,e,c,a$ constructed using the adjacency matrix. The same altering walk is colored blue in Figure \ref{fig: constructing_alternating_walks}. Agent labels are given in the first column and row of the table. To construct such a walk randomly we begin by choosing an agent at random. Here agent $j$ is chosen, with an ex ante probability of $\frac{1}{10}$ since there are ten agents in the network. Next we take an active step where one of agent $j$'s outlinks is chosen at random. Here we choose the outlink to agent $g$, an event with an ex ante probability of $\frac{1}{2}$ since agent $j$ has just two outlinks. Following the active step comes a passive step. In a passive step we move vertically to the row of an agent which does not direct a link to the current agent. Here we choose $a$ from the set $\left\{ a,b,c,d,e,f,i\right\}$  uniformly at random (i.e., with an ex ante probability of $\frac{1}{7}$). We continue with active and passive steps until we choose to stop or can proceed no further. Panel B reports the indegree and outdegree of each agent in the network. Observe that in active steps the probability of any feasible choice equals the inverse of the outdegree of the current agent. In passive steps the probability of any feasible choice equals the inverse of the number of nodes minus the indegree of the node chosen in the prior step minus 1 (since $i_{k}\neq i_{k+1}$). We can also construct alternating walks by the above procedure, but instead starting with a passive step. The shaded cells in the table shows which edges (ones) and non-edges (zeros) are in the walk.
}
\end{figure}

The length of an alternating walk equals the number of ordered dyads used to define it. An important type of alternating walk, which following \cite{Tao_NS16}, we call an \textit{alternating cycle}, is central to our algorithm.

\begin{definition}{\textsc{(Alternating Cycle)}}\label{def:alternating_cycle}
The alternating walk $C$ is an alternating cycle if $i_{1}=i_{l}$ and $C$ has even length.
\end{definition}

The length of an alternating cycle is at least four. Let $D_{i_{1}i_{2}},D_{i_{3}i_{1}},\ldots,D_{i_{l}i_{l-1}}$ be the sequence of adjacency matrix entries associated with alternating cycle $C$ in $D$. These entries necessarily form a sequence of zeros and ones (or ones and zeros).

Consider constructing an alternative digraph, say $\mathbf{D}'$, by replacing all the “ones” in the alternating cycle $C$ with “zeros” and all “zeros” with “ones”. Rewiring $\mathbf{D}$ in this way is degree preserving: $\mathbf{D}'$ has the same in- and out-degree sequence as $\mathbf{D}$. We refer to such operations as switching the cycle (since we switch the zeros and ones).

We use random alternating walks on the network in order to find alternating cycles. We then use these alternating cycles to rewire the network. This motivates the definition of what we call a schlaufe. A schlaufe is either an alternating walk which contains an alternating cycle (as the last part of the walk) or it is an alternating walk which cannot be continued. More precisely

\theoremstyle{definition}
\begin{definition}{\textsc{(Schlaufe)}}\label{def:schlaufe}
An alternating walk $H:=i_{1}i_{2}\ldots i_{l}$ is a \emph{schlaufe} if either \hfill\break
(i) There is a node $i_{k}\in\left\{ i_{1}i_{2}\ldots i_{l}\right\}$  with $k\neq l$ such that $i_{k}=i_{l}$ and $\left(k-l\right)\bmod2=0$. Furthermore for any two nodes $i_{j}$ and $i_{h}$ in $\left\{ i_{1}i_{2}\ldots i_{l-1}\right\}$  with $i_{j}=i_{h}$ and $j\neq h$ it holds that $\left(j-h\right)\bmod2=1$. \hfill\break
(ii) At node $i_{l}$ there is no other node $i_{l+1}$ such that the alternating walk could be extended with the unmarked link $\left(i_{l},i_{l+1}\right)$.
\end{definition}

In German schlaufe corresponds to “loop”, “bow” or “ribbon” (its plural is schlaufen); the latter translation is evocative of our meaning here. In the first case the schlaufe will coincide with an alternating walk which includes exactly one alternating cycle.\footnote{The requirement that $i_{k}=i_{l}$ and $\left(k-l\right)\bmod2=0$ ensures that $C=i_{k}i_{k+1}\ldots i_{l}$ is an alternating cycle (imposing even length). The “furthermore...” requirement ensures that if another node is visited multiple times it does not form an alternative cycle (imposing non-even length). See Figure \ref{fig: feasible_schlaufen_sequence} for an example.} Visually schlaufen of the first type, with the nodes appropriately placed, will look like loops and ribbons. In the second case the schlaufe does not include an alternating cycle.

Associated with a schlaufe, $R$, is a $K\times K$ violation matrix which records the number of extra links from group $k$ to group $l$ generated by switching the alternating cycle in $R$ (if there is one). Consider an alternating rectangle consisting of two boys and two girls. If initially one boy directs a link to the other and one girl directs a link to the other, then after switching the cycle the violation matrix will equal:\\
\begin{table}[h]
\centering
 \begin{tabular}{|c | c c |} 
 \hline
 Ego \textbackslash Alter & Boy & Girl \\
 \hline
 Boy & -1 & 1 \\ 
 Girl & 1 & -1 \\ [1ex] 
 \hline
 \end{tabular}
\end{table}
\\
After switching the cycle there are too few same gender links and too many mixed gender ones.

We call a sequence of schlaufen $\mathcal{R}=\left(R_{1},\ldots,R_{k}\right)$ feasible if (i) the cycles of the schlaufen are link disjoint and (ii) the sum of their violation matrices is zero (and for $i<k$ the sum of their violation matrices is not zero).

Conventional MCMC adjacency matrix re-wiring algorithms work by switching short cycles (e.g., alternating rectangles and compact alternating hexagons as in \cite{Rao_et_al_Sankhya96}). Switches of this type, while preserving the in- and out-degree sequence of the network will typically generate networks with the wrong inter-group link structure (i.e., non-zero link violation matrices). Our approach to solving this problem involves switching many alternating cycles simultaneously such that their individual link violation matrices sum to zero.

\subsection{The MCMC algorithm}

Let $\mathbf{S}=\mathbf{s}$ and $\mathbf{M}=\mathbf{m}$ be the degree sequence and cross link matrix of the network in hand. In order to a draw, say $\mathbf{D}'$, from $\mathbb{D}_{\mathbf{s},\mathbf{m}}$ we (i) start with a realization of $\left(\mathbf{s},\mathbf{m}\right)$, say $\mathbf{D}$, (ii) randomly construct (link disjoint) schlaufen, and (iii) switch any alternating cycles in them. While switching cycles will preserve the degree sequence, it may – as discussed earlier – result in a graph without the appropriate cross link matrix. In order to ensure that $\mathbf{D}'$ has the appropriate cross link matrix, we construct schlaufen until either the sum of their violation matrices equals zero or we stop randomly. If the sum of the schlaufen violation matrices is zero we move to $\mathbf{D}'$ from $\mathbf{D}$ by switching the cycles, otherwise we set $\mathbf{D}'=\mathbf{D}$. Proceeding in this way ensures that $\mathbf{D}'$ is, in fact, a random draw from $\mathbb{D}_{\mathbf{s},\mathbf{m}}$. After sufficiently many iterations of this process we show that a graph constructed in this way corresponds to uniform random draw from $\mathbb{D}_{\mathbf{s},\mathbf{m}}$. A formal statement of the procedure is provided by Algorithm \ref{alg: markov_draw}.

\begin{algorithm}
\caption{\textsc{Markov Draw Algorithm}}

\textbf{\underline{Inputs:}} An adjacency matrix $\mathbf{d}\in\mathbb{D}_{\mathbf{s},\mathbf{m}}$;
a mixing time $\tau$

\textbf{\underline{Procedure:}}
\begin{enumerate}
\item Set $t=0$.
\item With probability $1-q$ go to step 3, with probability $q$ go to
step 4.
\item find and mark a schlaufe (see Algorithm \ref{alg: Schlaufen_Detection}):
\begin{enumerate}
\item \textbf{if} the sum of the schlaufen violation matrices is zero, then 
\begin{enumerate}
\item switch the cycles in the schlaufen (changing the adjacency matrix
$\mathbf{d}$), 
\item unmark all links,
\item go to step 4.
\end{enumerate}
\item \textbf{else} 
\begin{enumerate}
\item with probability $\frac{1}{2}$, go to step 3 or 
\item with probability $\frac{1}{2}$, unmark all links and go to step 4.
\end{enumerate}
\end{enumerate}
\item Set $t=t+1$
\begin{enumerate}
\item \textbf{if} $t=\tau$ then return $\mathbf{d}$
\item \textbf{else} go to step 2
\end{enumerate}
\end{enumerate}
\textbf{\underline{Output:}} A uniform random draw $\mathbf{d}$ from $\mathbb{D}_{\mathbf{s},\mathbf{m}}$
\label{alg: markov_draw}
\end{algorithm}

Algorithm \ref{alg: markov_draw} uses a subroutine to find schlaufen. This subroutine, described in Algorithm \ref{alg: Schlaufen_Detection}, finds and marks a schlaufe in the graph.

\begin{algorithm}
\caption{\textsc{Schlaufen Detection Algorithm}}

\textbf{\underline{Inputs:}} An adjacency matrix $\mathbf{d}\in\mathbb{D}_{\mathbf{s},\mathbf{m}}$ (this
network may have marked links in it)

\textbf{\underline{Procedure:}}
\begin{enumerate}
\item Choose an agent/node, say $i$, at random.
\item Mark agent $i$ as active and
\begin{enumerate}
\item \textbf{if} feasible, randomly choose one of $i's$ (unmarked) outlinks,
say to $j$, and go to step 3;
\item \textbf{else} (i.e., no unmarked outlinks available) go to step 6.
\end{enumerate}
\item Mark edge $ij$, chosen in step 2 and
\begin{enumerate}
\item \textbf{if }agent $j$ is already marked passive, then go to step
6;
\item \textbf{else} go to step 4.
\end{enumerate}
\item Mark agent $j$, chosen in step 3, as passive and
\begin{enumerate}
\item \textbf{if} feasible, randomly choose an agent, say $k$, from among
those who \emph{do not} direct links to $j$, and go to step 5,
\item \textbf{else }go to step 6.
\end{enumerate}
\item Mark edge $kj$, with $k$ the agent chosen in step 4, as passive
and
\begin{enumerate}
\item \textbf{if} agent $k$ is already marked active, then go to step 6;
\item \textbf{else} go to step 2.
\end{enumerate}
\item return the (marked) adjacency matrix, the constructed schlaufe and
its violation matrix.
\end{enumerate}
\textbf{\underline{Output:}} A schlaufe, its violation
matrix and a marked adjacency matrix.
\label{alg: Schlaufen_Detection}
\end{algorithm}

To illustrate our method in more detail consider the network depicted in Panel A of Figure \ref{fig: feasible_schlaufen_sequence}. This network consists of two types of agents: gold (light) and blue (dark). The cross link matrix for the graph is given in Panel D. In Panels B and C a sequence of three schlaufen is shown. The first schlaufen is $R_{1}=jgabcdeca$. It is constructed through a sequence of active and passive steps as described earlier (see also the notes to Figure \ref{fig: constructing_alternating_walks} above). We begin by choosing agent $j$ randomly with a probability of $\frac{1}{10}$ (since there are ten agents in the network). We then take an active step, randomly choosing one of the two agents to which $j$ directs a link (i.e., either agent $g$ or $i$). Here we choose agent $g$. Next we take a passive step. Specifically we choose an agent at random from the set of agents that \emph{do not} direct a link to $g$ (the agent chosen in the previous active step). The probability associated with our choice in this passive step is $\frac{1}{7}$; this corresponds to the reciprocal number of agents in the network (i.e., $10$) minus the indegree the current agent (i.e., $2$) minus one (since self-loops are not allowed). We continue taking active and passive steps in this way until we visit $a$ for the second time. At this point we stop since our schlaufe now includes the alternating cycle $C_{1}=abcdeca$. Note that $c$ is also visited twice, but also that $cdec$ is not an alternating cycle since it is not of even length (see Definition \ref{def:alternating_cycle}). 

\begin{figure}
\caption{\label{fig: feasible_schlaufen_sequence}A feasible schlaufen sequence}

\begin{centering}
\includegraphics[scale=0.60]{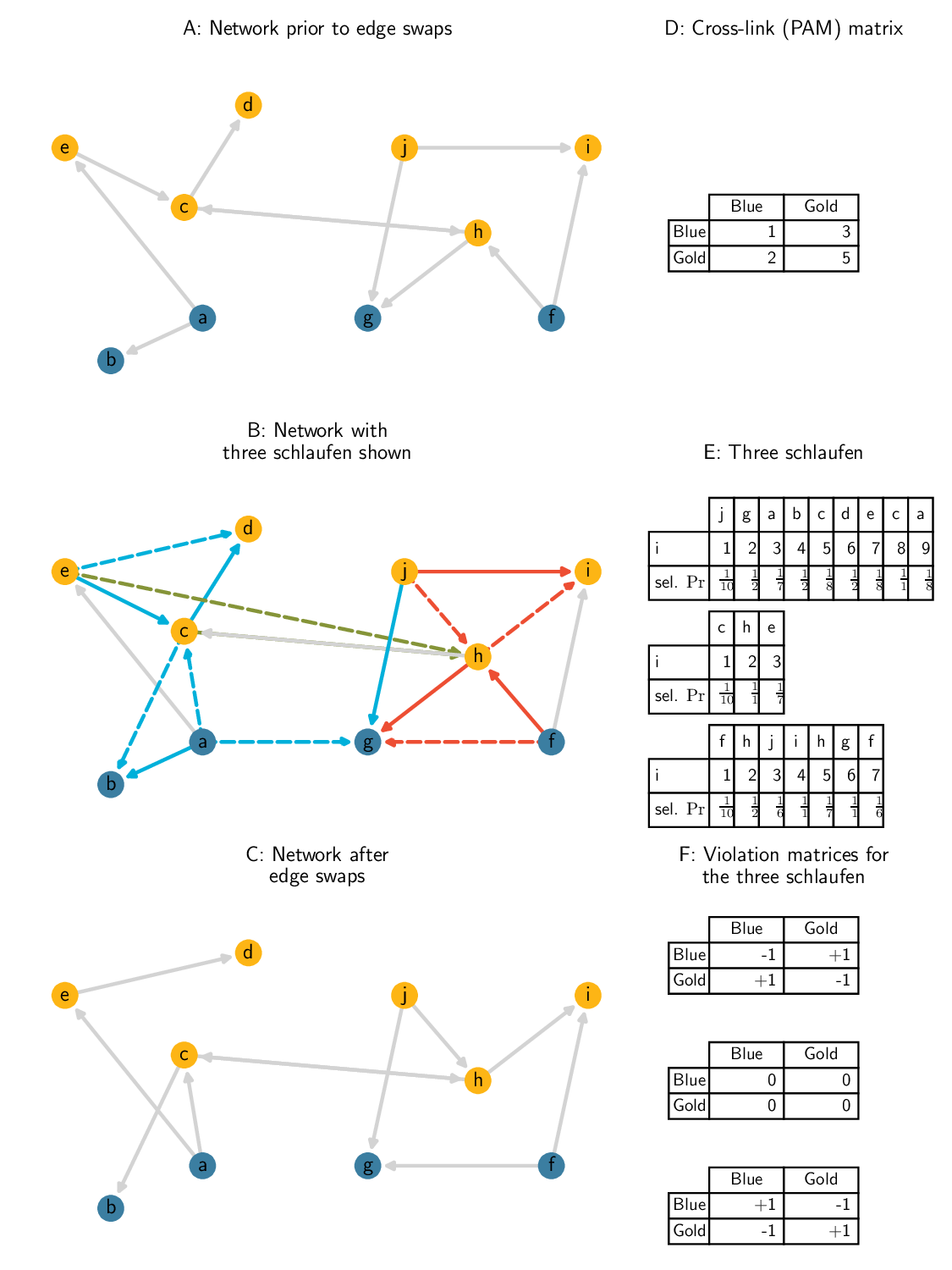}
\par\end{centering}
\caption*{
\underline{Source:} Authors' calculations. \hfill\break
\underline{Notes:} See the discussion in the main text. The figure depicts three link disjoint schlaufen with violation matrices which sum to zero. Panel E reports the (ex ante) probability that a given node was selected as the schlaufe was constructed. See equation \eqref{eq:schlaufen}.}
\end{figure}

As seen in the example we can calculate the probability of a schlaufe $R$ as we go through the algorithm (see Panel E). In Step 1 of Algorithm \ref{alg: Schlaufen_Detection} an agent is chosen with probability $\frac{1}{N}$. Next let $r_D^a(i)$ be the cardinality of the set of feasible out links in an active step. This set consists of all the out links of node $i$, which are not already marked in $\bf{D}$. Similarly, let $r_D^p(i)$ be the cardinality of the set of feasible outlinks in an passive step. That set consists of all the links $ij$ for which $ji$ is not in $\mathbf{D}$ and which are not already marked. The probability of $R = (i_1,..,i_l)$ can now be written as 
\begin{equation} \label{eq:schlaufen}
p_\mathbf{D}(R)= \frac{1}{N} \prod_{k=1}^{l-1}\left( \frac{1}{r_\mathbf{D}^a(i_k)}\left[k\bmod2\right] + \frac{1}{r_\mathbf{D}^p(i_k)}\left[(k-1)\bmod2\right]\right)
\end{equation}

In step 2 of Algorithm \ref{alg: markov_draw} we attempt to find a sequence of schlaufen with probability $1-q$ and do not change the adjacency matrix otherwise. In step 3, a schlaufen sequence $\mathcal{R} = (R_1, .. ,R_h)$ is constructed/found. After each detected schlaufe in this sequence, say $R_k$, any cycle in it is marked. Let $\mathbf{D}_k$ be the graph with the cycles of $R_1,..,R_{k-1}$ marked. After each schlaufe added the construction is stopped with probability $\frac{1}{2}$ . The probability of finding a cycle $R_k$ is $p_{\mathbf{D}_k}(R_k)$ as given in equation \eqref{eq:schlaufen} above. The total probability of a feasible schlaufen sequence $\mathcal{R}$ is therefore
\begin{equation} \label{eq:prob}
p_\mathbf{D}(\mathcal{R}) = (1-q)\frac{1}{2^{(h-1)}}\prod_{i=1}^h p_{\mathbf{D}_k}(R_k).
\end{equation}

\subsection{Correctness}

To show that our algorithm does indeed generate a uniform random draw from the set $\mathbb{D}_{\mathbf{s},\mathbf{m}}$ we use standard Markov chain theory (e.g., Chapters 7 and 10 of \cite{Mitzenmacher_Upfal_PC05}).

The random rewiring of the network implemented by Algorithm \ref{alg: markov_draw} can be described as a Markov chain. To show that, for $\tau$ large enough, it returns a uniform random draw from $\mathbb{D}_{\mathbf{s},\mathbf{m}}$ we prove that the stationary distribution of the Markov chain generated by Algorithm \ref{alg: markov_draw} is uniform on $\mathbb{D}_{\mathbf{s},\mathbf{m}}$. To show this it is helpful to develop a graphical representation of the Markov chain.

We denote the state graph of the Markov chain by $\Phi =
(\mathcal{V}_{\phi} , \mathcal{A}_{\phi})$. Its underlying vertex set $\mathcal{V}_\phi$ is the set of all elements in $\mathbb{D}_{\bf{s},\bf{m}}$. That is each node in our state graph is a network with degree sequence $\bf{S=s}$ and cross link matrix $\bf{M=m}$.
For network $\mathbf{D}$ in $\mathbb{D}_{\bf{s},\bf{m}}$, we denote by $v_{\mathbf{D}}$ the corresponding vertex in $\mathcal{V}_\phi$. The arc set $\mathcal{A}_\phi$ is defined as follows.
\begin{enumerate}
\item For all vertices we add the self loop $(v_{\mathbf{D}}, v_{\mathbf{D}})$ with (probability) weight $q$ (see Step 2 of Algorithm \ref{alg: markov_draw}).
\item  Let $\mathbf{D}$ and $\mathbf{D}'$ be two different networks in $\mathbb{D}_{\bf{s},\bf{m}}$. Let ${\mathbf{D}} \Delta {\mathbf{D}}'$ equal the union of the set of edges in ${\mathbf{D}}$, but not in ${\mathbf{D}}'$ and the set of edges in ${\mathbf{D}}'$, but not in ${\mathbf{D}}$. For each feasible  schlaufen-sequence $\mathcal{R}$, with cycle  edge set equal to ${\mathbf{D}} \Delta {\mathbf{D}}'$ we add the edge $(v_{\mathbf{D}}, v_{{\mathbf{D}}'})$ and assign to it probability weight $p_\mathbf{D}(\mathcal{R})$.
\item  Finally we add a directed loop $(v_{\mathbf{D}}, v_{\mathbf{D}})$ if the probability of all arrows leaving $v_{\mathbf{D}}$, introduced in points 1 and 2 immediately above, do not sum to 1. The probability of this loop is 1 minus the sum of the probability of all other outward arrows.
\end{enumerate}
The probability of any arc $a \in \mathcal{A}_\phi$ is denoted by $p(a)$. Note, by definition, the state graph can have parallel arcs and loops.

With these definitions in place we can prove correctness of the algorithm. First we show that the probability of the algorithm moving from graph $\mathbf{D}$ to $\mathbf{D}'$ coincides with the probability of moving in the reverse direction.
\begin{lemma} \label{lemma:symetrie}
For any two vertexes $v_{\mathbf{D}}, v_{\mathbf{D}'}$ the transition probability attached to $(v_{\mathbf{D}}, v_{\mathbf{D}'})$ equals that attached to $(v_{\mathbf{D}'}, v_{\mathbf{D}})$.
\end{lemma}
\begin{proof}
See appendix \ref{proof_symetrie}.
\end{proof}
Next we show the state graph is strongly connected. This means our Algorithm moves from any ${\mathbf{D}} \in \mathbb{D}_{\bf{s},\bf{m}}$ to any other ${\mathbf{D}'} \in \mathbb{D}_{\bf{s},\bf{m}}$ with positive probability.
\begin{lemma}\label{lemma:connected}
The state graph $\Phi$ is strongly connected.
\end{lemma}
\begin{proof}
See appendix \ref{connected}.
\end{proof}

With these two lemmata it is easy to show that the stationary distribution is uniform on $\mathbb{D}_{\bf{s},\bf{m}}$. This gives us the main result of the section.
\begin{theorem} \label{theorem:correctness}
Algorithm \ref{alg: markov_draw} is a random walk on the state graph $ \Phi $ which samples uniformly a network from $\mathbb{D}_{\bf{s},\bf{m}}$ for $\tau \rightarrow \infty$.
\end{theorem}
\begin{proof}
See appendix \ref{correctness1}.
\end{proof}

\begin{landscape}
\begin{figure}
\caption{Nyakatoke Village Risk-Sharing Network}
\begin{centering}
\includegraphics[scale=0.50]{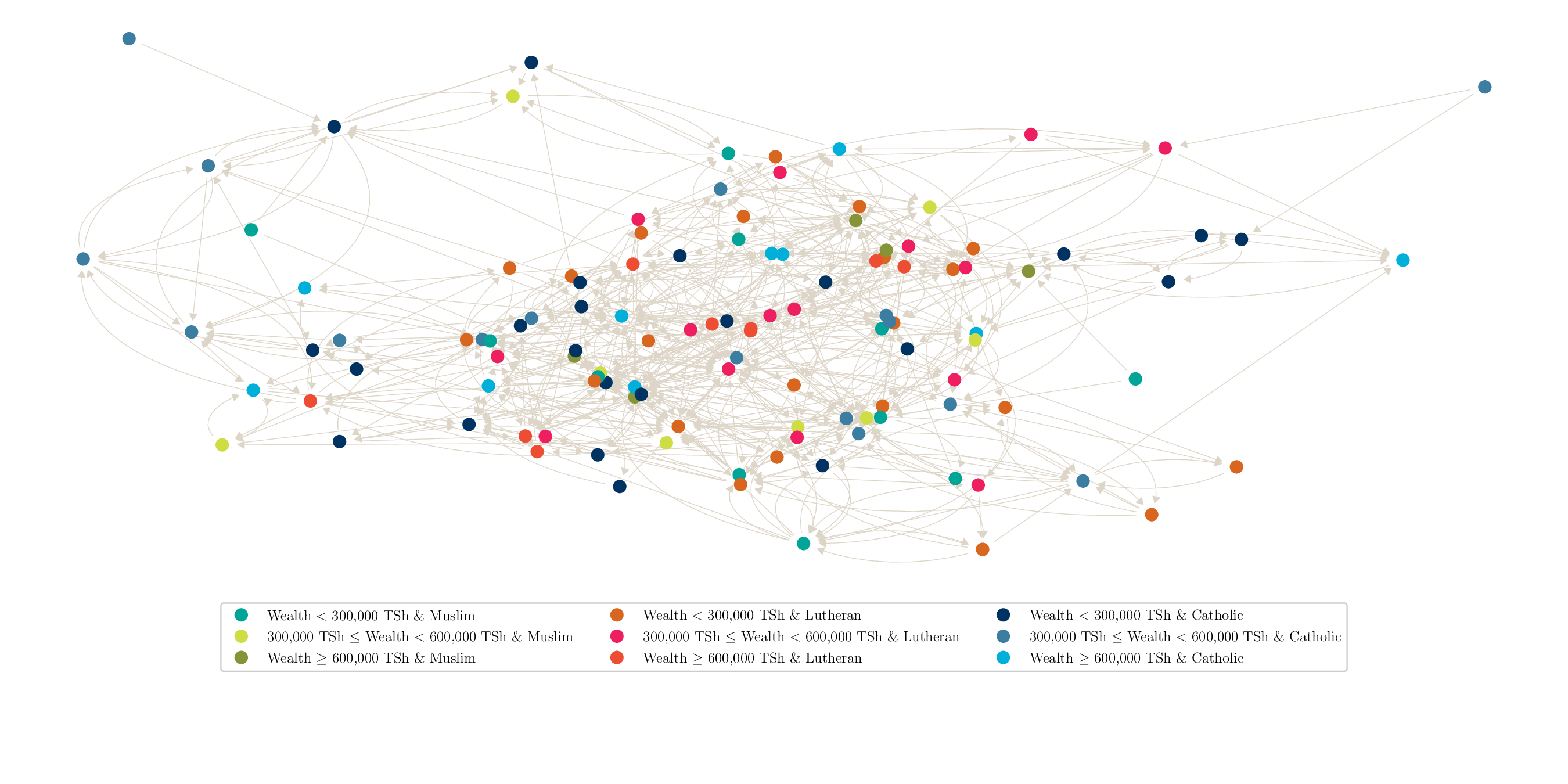}
\par\end{centering}
\caption*{\underline{Source:} \cite{deWeerdt_IAP04} and authors' calculations.\hfill\break
\underline{Notes:} Each household is colored according to their land and livestock wealth (measured in Tanzanian Shillings) and religion. The arrow head on the edges points to the ``alter" household with the link being sent by the tail ``ego" household.} 
\label{fig: fig_nyakatoke_network}
\end{figure}
\end{landscape}

\section{Application: risk-sharing links when agents value bridging capital} \label{sec:application}

\cite{deWeerdt_IAP04} studied the formation of risk-sharing links across 119 households in the rural village of Nyakatoke (located in Tanzania). He asked all adult individuals in the village who they could rely upon for help and, from their responses, constructed a network of directed links across households.\footnote{The prompt used by \cite{deWeerdt_IAP04} is suggestive of both mutuality and directionality, leading to some ambiguity in whether to interpret the collected edges as undirected or directed. \cite{Comola_Fafchamps_EJ2014} present evidence suggesting that the links given by households are directed. Specifically that they indicate to which other households they would turn to in the event of need. It is this interpretation that we give the links here.}. The resulting set of risk-sharing links is shown in Figure \ref{fig: fig_nyakatoke_network}.

Here we assess whether households value ``bridging capital", as suggested by \cite{Burt_StructuralHoles1995} and formalized in game-theoretic terms by \cite{Kleinberg_et_al_ACM2008} and others. 
If $k$ directs a link to $i$ but not to $j$, then $i$, by directing a link to $j$, may position herself to serve as a ``bridge" or ``broker" between $k$ and $j$. See Figure \ref{fig: Network_Benefit_Examples} above.

In the formal model of \cite{Kleinberg_et_al_ACM2008} agents gain utility from positioning themselves on length two paths connecting agents not directly connected themselves; however such utility gains are decreasing in the number of ``rival'' length two paths (i.e., those with other agents in the center). This suggest, for example, a network benefit function of
\begin{equation}
    g_{i}\left(\mathbf{d}\right)=\sum_{j\neq i,k,j}\sum_{k\neq i,j}\frac{D_{ki}D_{ij}\left(1-D_{kj}\right)}{\max\left(1,\sum_{l\neq j,k}D_{kl}D_{lj}\left(1-D_{kj}\right)\right)}.
\end{equation}
In this formulation any ``bridging" capital is shared equally across all agents $l$ on length two paths from $j$ to $k$ (with arc $jk$ absent). For example, if there are two bridging agents situated between $j$ and $k$, they each get half the benefit and so on. The marginal network benefit of edge $ij$ is thus
\begin{equation}
    s_{ij}\left(\mathbf{d}\right)=\sum_{k\neq i,j}\frac{D_{ki}D_{ij}\left(1-D_{kj}\right)}{\max\left(1,\sum_{l\neq j,k}D_{kl}D_{lj}\left(1-D_{kj}\right)\right)},
\end{equation}
from which the form of the locally best test follows.

From \cite{deWeerdt_IAP04} we also know that household land and livestock wealth, as well as religion (Catholic, Lutheran or Muslim), are important drivers of link formation in Nyakatoke. We divide households into three wealth bins, which in conjunction with religion, partitions households into nine groups; $X_i$ consists of the nine resulting group membership dummies with the 81 elements of $\Lambda$ parametrizing any homophily/heterophily across these groups. The remaining null model parameters are the $238=119 \times 2$ household-specific in- and out-degree heterogeneity parameters. This gives $\dim(\delta)=2 \times 119 + 9 \times 9$ = 319 null model ``nuisance" parameters. It is hard to imagine a testing approach with good properties in this setting which would not involve ``conditioning away" the null model parameter.

While the form of the locally best test statistic follows naturally from the form of the \cite{Kleinberg_et_al_ACM2008} network benefit function, it is less clear how to form a heuristic test with power to detect the alternative ``agents like to bridge disconnected groups". After some experimentation we settled on the difference between the 90th and 50th percentiles of the empirical distribution of betweenness-centrality across agents in the network as a suitable \emph{ad hoc} test statistic (other measures of dispersion give similar results). The intuition is that acquiring bridging capital is inherently rivalrous; the addition of links by other agents may reduce one's own network benefit. Competition to accumulate bridging capital should lead to more dispersion in betweenness-centrality across agents (than in a reference set of null model graphs). Winners of this competition (the 90th percentile) will have more bridging capital than the typical agent (the 50th percentile) in the network. We wish to emphasize that the ``ad hoc" descriptor of this statistic is apt. Indeed, an advantage of the formalism of an explicit network benefit function is that gives precision to the alternative of interest (in turn suggesting a suitable, in fact, optimal test statistic).

The left panel of Figure \ref{fig: nyakatoke_brokerage} plots simulation estimates of the distribution of the $90-50$ betweenness-centrality gap across three reference sets of networks: (i) \"{E}rdos-R\'{e}nyi graphs with the same number of links as observed in Nyakatoke, (ii) the set of all graphs with the same in- and out-degree sequences as observed in Nyakatoke, and (iii) the set of all graphs which additionally constrain the number of cross-group links to be the same as observed in Nyakatoke. The vertical line in the figure marks the value of the actual $90-50$ betweenness-centrality gap in Nyakatoke.

The three reference distributions in Panel A allow us to undertake three model adequacy tests: is Nyakatoke well-described by (i) the \"{E}rdos-R\'{e}nyi model, (ii) a directed $\beta$-model which places equal probability on all networks with the same in- and out-degree sequence as in Nyakatoke, or (iii) by the \cite{Charbonneau_EJ17} model described above? In all three cases we reject, but notice that as we enrich the null model the simulated reference distributions shift to the right.\footnote{The incremental effect of additionally controlling for homophily is modest.} Put differently a portion of the dispersion in betweenness-centrality across households observed in Nyakatoke is likely a by-product of degree heterogeneity and homophily. The rightward shifts in the reference distributions as we enrich the null model is indicative of how using a realistic null model may be important for avoiding spurious rejections in practice. That said, our decisive rejection of even the $319$ parameter \cite{Charbonneau_EJ17} model indicates that degree heterogeneity and wealth/religion homophily cannot explain \emph{all} of the inequality in betweenness-centrality we observe across agents in Nyakatoke.

\begin{figure}
\caption{Testing for bridging/brokerage preferences}
\begin{centering}
\includegraphics[scale=0.75]{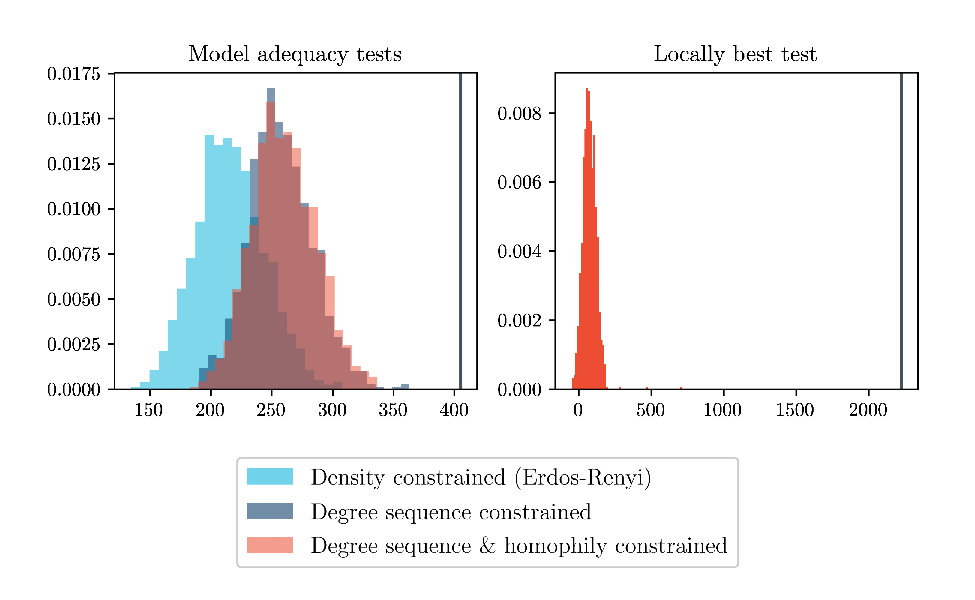}
\par\end{centering}
\caption*{\underline{Source:} \cite{deWeerdt_IAP04} and authors' calculations.\hfill\break
\underline{Notes:} Panel A presents MCMC estimates of the \emph{distribution} of the $90-50$ betweenness-centrality gap across agents for three reference sets of networks (as listed in the legend). Panel B shows the null distribution of the locally best test described in the main test. In this panel the reference set is all networks with the same in- and out-degree sequences and cross-link matrix as observed in Nyakatoke.} 
\label{fig: nyakatoke_brokerage}
\end{figure}

The right panel of Figure \ref{fig: nyakatoke_brokerage} plots the null distribution of the locally best test statistic for the alternative that households gain utility by bridging disconnected pairs of agents (as formalized by \cite{Kleinberg_et_al_ACM2008}). If we are willing to maintain that the true data generating process is either in the null or specified alternative model space, we can interpret a rejection as evidence \emph{for} $\gamma_0$ being positive. To implement this test we replace $\delta_0$ with its maximum likelihood estimate (MLE) computed under the null.\footnote{The computation of this MLE is described in detail by \cite{Dzemski_RESTAT18} and \cite{Yan_et_al_JASA18} and implemented in our Python package \textbf{ugd} for ``uniform graph draw".} As is clear from Panel B of Figure \ref{fig: nyakatoke_brokerage}, we decisively reject the null. 

Panel B is also suggestive of the power gains associated with the locally best test. If we were to standardize each of our test statistics using their respective reference distribution's mean and standard deviation, it is obvious that the locally best test statistic is more extremely positioned in the right tail of its null distribution (the Monte Carlo experiments reported in the Supplemental Web Appendix confirm the power advantages of the locally best test).

Using Algorithm \ref{alg: markov_draw} requires a choice of the mixing time parameter $\tau$. Although the mixing properties of our MCMC procedure are largely unexplored, we have found - by Monte Carlo experimentation -- that choosing $\tau$ such that each edge in the input graph is, on average, swapped at least once before the resulting output is considered a uniform random draw from the target set to yield acceptable results in practice. We use this approach here (also see the Python Juypter Notebook in the Supplemental Materials). The required value for $\tau$ is increasing in the dimension of the nuisance parameter $\delta$ and especially in the dimension of $\Lambda$. Hence the speed of the simulation algorithm declines in both $N$ and $K$.

\section{Limitations and future research} \label{sec:extensions}

The analysis in this paper, like much of the wider econometrics literature on games, is likelihood based. Our null model is fully parametric (albeit flexibly-so), while the alternative, due to the unmodeled NE selection function, is semiparametric. Under correct specification -- use case (ii) -- our test reveals whether $\gamma_0=0$ or  $\gamma_0>0$ (with a researcher-specified \emph{exact} Type I error rate, and a locally best Type II error rate). That is, we present a method for detecting whether agents form links ``strategically" in the presence of any pattern of homophily and degree heterogeneity allowed by the null.

It would be interesting to know whether detecting strategic interaction in the presence of \emph{arbitrary} homophily on observables and degree heterogeneity is possible. We know from the panel data literature that detecting state-dependence in the presence of heterogeneity is non-trivial and that modeling details matter \citep[e.g.,][]{Chamberlain_LALMD85}. Analogous questions arise here.

Our set-up assumes that researcher is able to a priori partition the support of agents' covariates into $K$ regions along which all homophilous sorting occurs. In practice this is an approximation. Developing data-based discretization rules (e.g., using clustering algorithms) and formalizing the nature of the approximations involved would be useful. It is possible that recent results on randomization inference by \cite{Canay_et_al_EM2017} could be useful for such an analysis. 

Key to our set-up is the exponential family structure (under the null) induced by the assumption of logistic random link-specific utility. While this is a strong assumption, it comes with considerable pay-off: we are (i) able to \emph{exactly} control size in (ii) the presence of a high dimensional nuisance parameter while (iii) also making no assumptions about equilibrium selection. Exponential family structures has proved highly fruitful in other areas of econometrics; applications in panel data being most closely connected to the present setting. Our similarity and local optimality results build on classic results in the theory of testing in exponential families (e.g., \cite{Ferguson_MS67} and \cite{Lehmann_Romano_TSH05}).

While obvious, and generic to most testing problems, it is important to understand that our test may have low power in some directions (in extreme cases even power equal to size). As an example imagine agents gain utility from linking with popular agents (as in preferential attachment models), such that $g_{i}\left(\mathbf{d}\right)=\sum_{j\neq i}d_{ij}\left[\sum_{k\neq i}d_{kj}\right]$. This model yields $s_{ij}\left(\mathbf{d}\right)=\sum_{k\neq i}d_{kj}$, which is almost equal to the indegree of agent $j$. Hence the distribution of $s_{ij}\left(\mathbf{D}\right)$ across $\mathbb{D}_{\mathbf{s},\mathbf{m}}$ will be nearly degenerate. Examples of this type are not unique to our setting. See \cite{Lehmann_Romano_TSH05} for general impossibility results.

Finally, while we are able to prove that our simulation algorithm works for $\tau$ ``large enough", we don't currently have a formal handle on the mixing properties of our procedure. This is not just a limitation of our work, but of much of the related work in the discrete math and computer science literature (e.g., \cite{Cooper_Comb_2007} and \cite{Erdosr_Comb_2018}). Our limited simulation experiments suggest relatively fast mixing. \footnote{A simple heuristic is to increase $\tau$ until one's results are not sensitive to further increases in it.}

These limitations notwithstanding, we nevertheless see potential for the widespread use of the methods presented in this paper in empirical social and economic network research (and, with modification, in other settings where strategic interaction is important). We hope that the ability to easily embed formal game-theoretic models of network formation of the type surveyed by, for example, \cite{Jackson_NetBook08} and \cite{Goyal_Book2021}, into heterogeneity-rich dyadic linking models will be attractive to empirical researchers. While not emphasized here, we also expect our simulation algorithm to find use in other settings where binary matrix simulation is an important part of researchers' toolkits \citep[e.g.,][]{Gotelli_EC2000}. Finally our focus on score type tests may represent a fruitful direction for further research on testing in incomplete models \citep[e.g.,][]{Chen_Kaido_WP2021}.

The Supplemental Web Appendix shows how to adapt our results to bi-partite networks. There we show how ideas in this paper might be used to, for example, study airline entry into different routes as in \cite{Ciliberto_Tamer_EM2009}. The set-up allows for complex airline preferences over their own route map as well as how they vary with the route maps chosen by their competitors. We also shows how our simulation algorithm can be used for more traditional conditional likelihood estimation and inference problems. A carefully annotated Python Jupyter, Notebook illustrating how the methods in this paper work in practice, is  available in the Supplemental Materials.

\bibliography{StrategicInteractionTesting.bib}

\newpage
\setcounter{page}{0}
\pagenumbering{arabic}
\setcounter{page}{1}
\appendix
\section{Supplemental Web Appendix}
The appendix includes proofs of the theorems stated in the main text as well as statements and proofs of supplemental lemmata. All notation is as established in the main text unless stated otherwise. Equation numbering continues in sequence with that established in the main text. 

In addition to proofs, Section \ref{app: monte_carlo} of the Appendix summarizes the results of a small set of Monte Carlo experiments and Section \ref{app: applications} discusses additional applications of our MCMC algorithm.

\subsection{Measurability of the likelihood} \label{measurablity}

For the equation \eqref{eq:likelihood} to be well defined we must show that $\mathcal{N}(\mathbf{d},\cdot ;\theta)$ is measurable. For any network $\mathbf{d}$ we can define a function $\mathcal{N}(\mathbf{d},\cdot;\theta)$, which assigns to the realization $\bf{U}=\bf{u}$ a probability weight for the pure strategy which corresponds to $\mathbf{d}$. We now show that there is a measurable function $\mathcal{N}(\mathbf{d},\cdot;\theta)$  satisfying these conditions. 
\\
\\
Observe that every realization $\mathbf{u} \in \mathbb{R}^n$ of the taste shock $\bf{U}$  corresponds to a game in normal/strategic form ($n = N \times (N-1)$ equals the number of random utility shocks in the network formation game). Every game in normal form has a set of Nash equilibria. We  define the set valued function $X:\mathbb{R}^n \rightarrow \{ \sigma | \sigma \subset \Sigma \} $, which assigns to each taste shock $\bf{u} \in \mathbb{R}^n$ the set of Nash equilibria in the game defined by (corresponding to) $\bf{u}$. We next use results from the theory of random sets, as outlined in \cite{Molchanov_RandomSets2017}, to show that there is a measurable equilibrium selection function. 
\\
\\
The theory of random sets analyzes set valued random variables. We want to apply this theory to $X$, the set of NE associated with our game when $\bf{U}=\bf{u}$. In order to do so we have to show measurability of $X$. 
\\
\\
Note the NE are the solutions to a system of $N \times 2^{N-1}$ inequalities, $2^{N-1}$ inequalities for each of the $N$ players. 
\\
\\
We consider one inequality, say, $i$, and claim that the set of mixes strategies which fulfils this inequality, say, $X_i$, is a random set. On either side of the inequality is a convex combination of the payoff entries of the normal form table. The weights of the convex combination are determined by the mixed strategies, each weight is the product of the corresponding mixed strategy weight of each player. The entries of the normal form table depend on the utility function. Importantly, the payoff entries are continuous in the random utility taste shocks. Now consider an arbitrary mixed strategy $\sigma$; the distance function
\[ \rho(\sigma,X_{i})=\inf\{\left.\left\Vert \sigma-x\right\Vert \thinspace\right|\thinspace x\in X_{i}\}\]
is continuous in the taste shock, because the payoff entries of the convex combination are continuous in the taste shock. By statement (iv) of Molchanov's (\citeyear{Molchanov_RandomSets2017}) Fundamental Measurability Theorem for Multifunctions (Theorem 1.3.3 on p. 59) $X_i$ is measurable. \\
\\ 
Next observe that the set of NE coincides with the intersection of the solution sets for each inequality
\[X = \cap_i X_i.\]
By part (iv) of Molchanov's (\citeyear{Molchanov_RandomSets2017}) theorem on the Measurability of Set-Theoretic Operations (Theorem 1.3.25 on p. 69) $X$ is measurable and therefore a random set. 
\\
\\
We know from Nash's existence theorem that for each game there exists a NE. Therefore $X$ is nonempty.  The set of NE is a closed set. We therefore can apply Molchanov's (\citeyear{Molchanov_RandomSets2017}) Fundamental Selection Theorem (Theorem 1.4.1 on p. 77) and find a measurable selection $\xi: \mathbb{R}^n \rightarrow \Sigma$ which assigns to each taste shock $\mathbf{u}$, a NE.
\\
\\
Let $h_{\mathbf{d}}: \Sigma \rightarrow [0,1]$ be the function which assigns to every mixed strategy the probability of $\mathbf{d}$ by multiplying the mixed strategies weights corresponding to $\mathbf{d}$. Since multiplication is a measurable operation $h_\mathbf{d}$ is measurable and $ \mathcal{N}(\mathbf{d},\cdot;\theta):= h_{\mathbf{d}} \circ \xi$ satisfies the desired properties.
\\

\subsection{Derivation of locally best test statistic} \label{app: derivative_proof}

We begin with a high level overview of our argument; with a formal proof immediately following.

One feature of $s_{ij}(\mathbf{d})$, which will prove central to our analysis, is that is has finite range. To see this observe that since the set of all networks $\mathbb{D}_{N}$ is finite, $s_{ij}\left(\mathbf{d}\right)$ also takes only a finite number of values. Let $\mathbb{S}=\left\{ \underline{s},s_{1},\ldots,s_{M},\overline{s}\right\}$ be the set of possible values for $s_{ij}\left(\mathbf{d}\right)$, ordered from smallest to largest.

An example illustrates. If $s_{ij}(\mathbf{d})=d_{ji}$, as occurs when agents prefer reciprocated links, we have $\mathbb{S}=\left\{0,1\right\}$. If $s_{ij}(\mathbf{d})=\sum_{k}d_{ki}d_{kj}$, as when agents prefer supported links, we have $\mathbb{S}=\left\{0,1,\dots,N-2\right\}$. Finiteness of the cardinality of $\mathbb{S}$ (for a given $N$) plays an important role in our analysis, as will become apparent below.

To understand the likelihood \eqref{eq:likelihood} it is helpful to consider a (relatively) simple example. This example will also help in understanding our derivation of the optimal test statistic below. Assume that $s_{ij}(\mathbf{d})=d_{ji}$ such that agents prefer reciprocated links when $\gamma>0$. In this example $s_{ij}(\mathbf{d})$ equals either zero ($j$ does not reciprocate) or one ($j$ does reciprocate). We can use the two elements of $\mathbb{S}$ to partition the real line into what we will call \textit{buckets}:
\begin{equation}\label{eq:buckets_for_reciprocity}
\mathbb{R}=\left(-\infty,\mu_{ij}\right]\cup\left(\mu_{ij},\mu_{ij}+\gamma\right]\cup\left(\mu_{ij}+\gamma,\infty\right).    
\end{equation}
Here $\mu_{ij}=A_{i}+B_{j}+X_{j}'\Lambda_{0}X_{i}$ equals the systematic component of baseline utility generated by arc $ij$. Next consider the realization of $U_{ij}$, the idiosyncratic utility agent $i$ gets when she directs a link to $j$. If $U_{ij}$ falls into the first bucket in \eqref{eq:buckets_for_reciprocity}, then agent $i$ will always direct a link to $j$; irrespective of whether $j$ chooses to direct a link to $i$ or not. If $U_{ij}$ falls into the middle or \textit{inner} bucket, however, then $i$ will direct a link to $j$ only if $j$ reciprocates. Finally, if $U_{ij}$ falls into the last bucket, then $i$ will never direct a link to $j$ regardless of whether $j$ directs a link to $i$ or not. We will call the first and last buckets in \eqref{eq:buckets_for_reciprocity} \textit{outer} buckets.

If both $U_{ij}$ and $U_{ji}$ fall in their respective \textit{inner} buckets, then the $\{i,j\}$ dyad can either take the empty ($D_{ij}=D_{ji}=0$) or reciprocated ($D_{ij}=D_{ji}=1$) configuration in equilibrium. In contrast, if \textit{either} $U_{ij}$ or $U_{ji}$ falls into an outer bucket, then the $\{i,j\}$ dyad's wiring is uniquely determined. For example if $U_{ij}$ is in the first outer bucket and $U_{ji}$ is in the inner bucket, then the $\{i,j\}$ dyad will take the reciprocated form with probability one. It is a strictly dominant strategy for $i$ to direct an link to $j$ in this case and a best response for $j$ to reciprocate.

For $\mathbf{U}=\mathbf{u}$ and $\theta=\theta_0$, let $J(\mathbf{u};\theta_0)\leq{\binom{N}{2}}$ equal the number of dyads $\{i,j\}$, where both $u_{ij}$ and $u_{ji}$ fall into their inner bucket. For each of these dyads both the empty and reciprocated configuration is an equilibrium outcome. There are therefore $2^{J(\mathbf{u};\theta_0)}$ equilibrium networks in this case; the $\mathcal{N}(\mathbf{d},\mathbf{u};\theta_0)$ function would assign some probability between zero and one to each of these $2^{J(\mathbf{u};\theta_0)}$ networks (summing to one in total).

Recall that $\mathbb{S}=\left\{ \underline{s},s_{1},\ldots,s_{M},\overline{s}\right\}$ equals the possible values of $s_{ij}(\mathbf{d})$, arranged from smallest to largest. We can use these support points to partition $\mathbb{R}$ into a set of intervals $\mathbb{B}$:
\begin{multline}\label{eq:buckets}
\mathbb{R=}	\left(-\infty,\mu_{ij}+\gamma\underline{s}\right]\cup\left(\mu_{ij}+\gamma\underline{s},\mu_{ij}+\gamma s_{1}\right]\cup \\
	\cdots\cup\left(\mu_{ij}+\gamma s_{M},\mu_{ij}+\gamma\overline{s}\right]\cup\left(\mu_{ij}+\gamma\overline{s},\infty\right).    
\end{multline}
The elements of $\mathbb{B}$, called \textit{buckets}, correspond to the intervals listed in \eqref{eq:buckets}. In principle we should write $\mathbb{B}_{ij}$ instead of $\mathbb{B}$, reflecting the dependence of the bucket definitions on the value of $\mu_{ij}$, the systematic non-strategic utility associated with an $i$-to-$j$ link. However, since this dependence is not essential to any of the arguments that follow we leave it implicit. Note that the cardinality of $\mathbb{B}$ does not depend on $\mu_{ij}$, but instead equals $|\mathbb{S}|+1$.

Agent $i$'s linking behavior vis-a-vis $j$ depends on which bucket $U_{ij}$ falls into. For $B\in\mathbb{B}$, if $U_{ij}\in B$, then we say $U_{ij}$ is in, or falls into, bucket $B$. The first and last buckets, respectively $\left(-\infty,\mu_{ij}+\gamma\underline{s}\right]$ and $\left(\mu_{ij}+\gamma\overline{s},\infty\right)$, play an important role in our argument. We call these two buckets \textit{outer buckets}. The rest of the buckets we call \textit{inner buckets}.

If $U_{ij}$ falls into one of these outer buckets then player $i$ has a pure strategy for $d_{ij}$ which is strictly dominating. Specifically if $U_{ij}$ falls into the lowest bucket, then $i$ will direct an link to $j$ regardless of what actions are taken by the other agents in the network. The marginal utility generated by link $ij$ is so large that it remains positive across all possible configurations of the rest of the network; hence $i$ always chooses to direct an link to $j$.

If, instead, $U_{ij}$ falls into the highest bucket, then $i$ will never direct an link to $j$. In this case the marginal utility associated with link $ij$ is so low that it remains negative across all possible configurations of the rest of the network; hence $i$ never chooses to direct a link to $j$.

Finally, if $U_{ij}$ falls into an inner bucket, say $\left(\mu_{ij}+\gamma s_{m},\mu_{ij}+\gamma s_{m+1}\right]$, then agent $i$'s optimal choice for $d_{ij}$ is contingent upon the linking behavior of other agents. If other agents' link actions are such that $s_{ij}(\mathbf{d}) \geq s_m$, then it is a best response for $i$ to link with $j$, but not otherwise.

The vector of idiosyncratic taste shocks, $\mathbf{U}$ contains $n=N(N-1)$ elements; one for each possible arc. Let the boldface subscripts $\mathbf{i=1,2,}\ldots$ index these potential arcs in arbitrary order (e.g., $\mathbf{i}$ maps to some $ij$ and vice-versa). Let $\mathbf{b}\in\mathbb{B}^{n}\overset{def}{\equiv}\mathbb{B}\times\cdots\times\mathbb{B}$ and $\mathbf{U}=\left(U_{\mathbf{1}},\ldots,U_{\mathbf{n}}\right)'$; we have that $\mathbf{U}\in\mathbf{b}$ for $\mathbf{b}\in\mathbb{B}^{n}$ so that each element of $\mathbf{u}$ falls into a bucket.

With the above notation established we can rewrite the likelihood \eqref{eq:likelihood} as:

\begin{equation}\label{eq:likelihood_bucket_decomp}
P\left(\mathbf{d};\theta,\mathcal{N}\right)=\sum_{\mathbf{b}\in\mathbb{B}^{n}}\int_{\mathbf{u}\in\mathbf{b}}\mathcal{N}\left(\mathbf{d},\mathbf{u};\theta\right)f_{\mathbf{U}}\left(\mathbf{u}\right)\mathrm{d}\mathbf{u}
\end{equation}
Expression \eqref{eq:likelihood_bucket_decomp} suggests a derivation by cases approach to finding $\left.\frac{\partial P\left(\mathbf{d};\theta,\mathcal{N}\right)}{\partial\gamma}\right|_{\gamma=0}$. Fortunately a brute force exhaustive approach is not required because it is possible to show that most of the summands in \eqref{eq:likelihood_bucket_decomp} do not influence the derivative at $\gamma=0$.

Let $\tilde{\mathbb{B}}^{n}$ be the set of bucket configurations with at least two inner buckets. If at least two elements of $\mathbf{U}$ fall in inner buckets, then we have that $\mathbf{U}\in\mathbf{b}$ with $\mathbf{b}\in\tilde{\mathbb{B}}^{n}$. If, instead, at most one element of $\mathbf{U}$ falls in an inner bucket, then we have that $\mathbf{U}\in\mathbf{b}$ with
$\mathbf{b}\in\mathbf{\mathbb{B}}^{n}\setminus\tilde{\mathbb{B}}^{n}$. This set-up gives the likelihood decomposition:
\begin{equation}
    P\left(\mathbf{d};\theta,\mathcal{N}\right)=\tilde{P}\left(\mathbf{d};\theta,\mathcal{N}\right)+Q\left(\mathbf{d};\theta,\mathcal{N}\right),
\end{equation}
with
\begin{align}
\tilde{P}\left(\mathbf{d};\theta,\mathcal{N}\right) &=	\sum_{\mathbf{b}\in\mathbf{\mathbb{B}}^{n}\setminus\tilde{\mathbb{B}}^{n}}\int_{\mathbf{u}\in\mathbf{b}}\mathcal{N}\left(\mathbf{d},\mathbf{u};\theta\right)f_{\mathbf{U}}\left(\mathbf{u}\right)\mathrm{d}\mathbf{u} \\
Q\left(\mathbf{d};\theta,\mathcal{N}\right) &=	\sum_{\mathbf{b}\in\tilde{\mathbb{B}}^{n}}\int_{\mathbf{u}\in\mathbf{b}}\mathcal{N}\left(\mathbf{d},\mathbf{u};\theta\right)f_{\mathbf{U}}\left(\mathbf{u}\right)\mathrm{d}\mathbf{u}.
\end{align}

To proove Theorem \ref{thm: derivative} we show that for $\gamma\rightarrow0$
\begin{equation}
P\left(\mathbf{d};\theta,\mathcal{N}\right)=\tilde{P}\left(\mathbf{d};\theta,\mathcal{N}\right)+\mathcal{O}\left(\gamma^{2}\right).    
\end{equation}
Intuitively, this follows from the fact that the chance that two or more elements of $\mathbf{U}$ fall in inner buckets is negligible when $\gamma$ is close to zero (because most of the probability mass for $U_{ij}$ is contained in the two outer buckets when strategic interactions are small). Hence when calculating the optimal test statistic we are free to focus on the cases where either all, or all but one, of the elements of $\mathbf{U}$ fall in outer buckets. We can then show that
\begin{equation}\label{eq:derivative_equality}
    \left.\frac{\partial P\left(\mathbf{d};\theta,\mathcal{N}\right)}{\partial\gamma}\right|_{\gamma=0}=\left.\frac{\partial\tilde{P}\left(\mathbf{d};\theta,\mathcal{N}\right)}{\partial\gamma}\right|_{\gamma=0}.
\end{equation}
Hence to derive the form of $\left.\frac{\partial P\left(\mathbf{d};\theta,\mathcal{N}\right)}{\partial\gamma}\right|_{\gamma=0}$ we need only calculate $\left.\frac{\partial\tilde{P}\left(\mathbf{d};\theta,\mathcal{N}\right)}{\partial\gamma}\right|_{\gamma=0}.$ This calculation is non-trivial, but doable. Details of this calculation are provided in the proof.

\subsubsection*{Preliminary results}
\begin{lemma}\label{lemma:derivative_helper}
Any differentiable function $f\in\mathcal{O}\left(\gamma^{2}\right)$ with $f\left(0\right)=0$ has a derivative of zero at point zero.
\end{lemma}

\begin{proof}
For $f\in\mathcal{O}\left(\gamma^{2}\right)$ we have, for some $C>0$ and $\epsilon>0$, that
\begin{equation}
\left|f\left(\gamma\right)\right|<C\gamma^{2}    
\end{equation}
for all $\gamma\in\left[-\epsilon,\epsilon\right]$. The derivative of $f$ at $\gamma=0$ equals
\begin{equation}
f'\left(0\right)
=\underset{\gamma\rightarrow0}{\lim}\frac{f\left(\gamma\right)-f\left(0\right)}{\gamma}
=\underset{\gamma\rightarrow0}{\lim}\frac{f\left(\gamma\right)}{\gamma},    
\end{equation}
with the second equality because $f\left(0\right)=0$. As $\gamma\rightarrow0$, we will have $\gamma<\epsilon$ so that
\begin{equation}
f'\left(0\right)
=\underset{\gamma\rightarrow0}{\lim}\frac{f\left(\gamma\right)}{\gamma}\leq\underset{\gamma\rightarrow0}{\lim}\frac{C\gamma^{2}}{\gamma}
=\underset{\gamma\rightarrow0}{\lim}C\gamma    
\end{equation}
which goes to zero as $\gamma\rightarrow0$ as claimed.
\end{proof}


\subsubsection*{Proof of Theorem \ref{thm: derivative}} 

We begin with the likelihood decomposition \eqref{eq:likelihood_bucket_decomp} given above. The number of summands in \eqref{eq:likelihood_bucket_decomp} depends on the partition that $s_{ij}(\mathbf{d})$ induces on $\mathbb{R}$. For a positive $\gamma$, the number neither depends on the exact value of $\gamma$, nor on the other covariates and parameters. Intuitively, as long as $\gamma$ is positive, there is a positive probability that $\mathbf{U}$ falls in any combination of buckets. The number of summands in \eqref{eq:likelihood_bucket_decomp} is typically large. 
The buckets $\mathbf{b}$ of $\mathbb{B}^{n}$ and the function $\mathcal{N}$ depend on $\gamma$. 

We have that
\begin{align*}
\frac{\partial P\left(\mathbf{d};\theta,\mathcal{N}\right)}{\partial\gamma}&=\frac{\partial}{\partial\gamma}\left\{ \sum_{\mathbf{b}\in\mathbb{B}^{n}}\int_{\mathbf{u}\in\mathbf{b}}\mathcal{N}\left(\mathbf{d},\mathbf{u};\theta\right)f_{\mathbf{U}}\left(\mathbf{u}\right)\mathrm{d}\mathbf{u}\right\} \\
&=\sum_{\mathbf{b}\in\mathbb{B}^{n}}\frac{\partial}{\partial\gamma}\int_{\mathbf{u}\in\mathbf{b}}\mathcal{N}\left(\mathbf{d},\mathbf{u};\theta\right)f_{\mathbf{U}}\left(\mathbf{u}\right)\mathrm{d}\mathbf{u},
\end{align*}
The switching of summation and derivative operator is possible because the number of summands does not depend on $\gamma$. We could try to take the derivative of each summands integral boundaries and of  $\mathcal{N}(\mathbf{d},.;\theta)$.  But there is no need to boil the ocean, because regardless of $\mathcal{N}(\mathbf{d},.;\theta)$ most of the summands are 0.
To show this we consider three sets of summands.

\subsubsection*{Case 1: more than two buckets in $\mathbf{B}$ are inner buckets}
Recall that the boldface subscripts $\mathbf{i=1,2,}\ldots$ index the $n=N\left(N-1\right)$ directed dyads in arbitrary order. Consider a set of buckets $\mathbf{b}$ where two or more of them are inner buckets. Without loss of generality assume that the $L\geq2$ inner buckets correspond to $b_{\mathbf{1}},\ldots,b_{\mathbf{L}}$ of $\mathbf{b}=\left(b_{\mathbf{1}},\ldots,b_{\mathbf{n}}\right)$. The shape of the $\mathbf{l}^{th}$ bucket is $\left(\gamma\underline{s}_{\mathbf{l}},\gamma\bar{s}_{\mathbf{l}}\right]$ with $\underline{s}_{\mathbf{l}}<\bar{s}_{\mathbf{l}}$ coinciding with the bucket borders induced by the precise form of strategic interaction specified under the alternative. We normalize the dyad-specific systematic utility component $\mu_{ij}=0$ without loss of generality.

Recall that $\tilde{\mathbb{B}}^{n}$ is the set of bucket configurations with two or more inner buckets. For any $\mathbf{b}\in\tilde{\mathbb{B}}^{n}$ we can derive the upper bound:

\begin{align*}
\int_{\mathbf{u}\in\mathbf{b}}\mathcal{N}\left(\mathbf{d},\mathbf{u};\theta\right)f_{\mathbf{U}}\left(\mathbf{u}\right)\mathrm{d}\mathbf{u}	    &=\int_{\mathbf{u}\in\mathbf{b}}\mathcal{N}\left(\mathbf{d},\mathbf{u};\theta\right)\left[\prod_{\mathbf{i}}f_{U}\left(u_{\mathbf{i}}\right)\right]\mathrm{d}\mathbf{u} \\
	&\leq\int_{\gamma\underline{s}_{\mathbf{1}}}^{\gamma\overline{s}_{\mathbf{1}}}f_{U}\left(u_{\mathbf{1}}\right)\times\cdots\times\int_{\gamma\underline{s}_{\mathbf{L}}}^{\gamma\overline{s}_{\mathbf{L}}}f_{U}\left(u_{\mathbf{L}}\right)\int_{\mathbf{u}_{-L}\in\mathbf{b}_{-L}}f_{\mathbf{U}_{-L}}\left(\mathbf{u}_{-L}\right)\mathrm{d}\mathbf{u} \\
	&<\int_{\gamma\underline{s}_{\mathbf{1}}}^{\gamma\overline{s}_{\mathbf{1}}}f_{U}\left(u_{\mathbf{1}}\right)\times\cdots\times\int_{\gamma\underline{s}_{\mathbf{L}}}^{\gamma\overline{s}_{\mathbf{L}}}f_{U}\left(u_{\mathbf{L}}\right)\mathrm{d}u_{\mathbf{1}}\cdots\mathrm{d}u_{\mathbf{L}} \\
	&<\int_{\gamma\underline{s}_{\mathbf{1}}}^{\gamma\overline{s}_{\mathbf{1}}}1\times\cdots\times\int_{\gamma\underline{s}_{\mathbf{L}}}^{\gamma\overline{s}_{\mathbf{L}}}1\mathrm{d}u_{\mathbf{1}}\cdots\mathrm{d}u_{\mathbf{L}} \\
	&=\gamma^{L}\left(\overline{s}_{\mathbf{1}}-\underline{s}_{\mathbf{1}}\right)\times\cdots\times\left(\overline{s}_{\mathbf{L}}-\underline{s}_{\mathbf{L}}\right)    
\end{align*}
where $\mathbf{u}_{-L}$ denotes the vector $\mathbf{u}$ after removal of its first $L$ components and similarly for $\mathbf{b}_{-L}$. The first equality follows from independence of the components of $\mathbf{u}$, the second (weak) inequality from the fact that $\mathcal{N}\left(\mathbf{d},\mathbf{u};\theta\right)\leq1$ for all $\mathbf{u}\in\mathbb{U}$. The third (strict) inequality follows because $f_{\mathbf{U}_{-L}}\left(\mathbf{u}_{-L}\right)$ is a density and the integration is not over all of $\mathbb{R}^{n-L}$. The fourth (strict) inequality arises because when $f_{U}\left(u\right)$ is the logistic density we have that $f_{U}\left(u\right)=F_{U}\left(u\right)\left[1-F_{U}\left(u\right)\right]<1$ for all $u$ on a compact interval of the real line. We conclude that any summand where $\mathbf{b}$ has two or more inner buckets is $\mathcal{O}\left(\gamma^{2}\right)$ for $\gamma\rightarrow0$.

We have, directly from this argument, that $Q\left(\mathbf{d};\theta,\mathcal{N}\right)\in\mathcal{O}\left(\gamma^{2}\right)$ and furthermore that $Q\left(\mathbf{d};\left(0,\delta'\right)',\mathcal{N}\right)=0$ (since inner buckets have zero probability when $\gamma=0$). Hence, by Lemma \ref{lemma:derivative_helper}, we have that
\begin{equation*}
    \left.\frac{\partial Q\left(\mathbf{d};\theta,\mathcal{N}\right)}{\partial\gamma}\right|_{\gamma=0}=0.
\end{equation*}
This is enough to show equation \eqref{eq:derivative_equality} above. This simplification is essential to the overall result, as it allows us to proceed without knowing any details about the form of the equilibrium selection rule $\mathcal{N}$ when $\mathbf{U}$ takes values which admit multiple NE networks.

\subsubsection*{Case 2: No bucket in $\mathbf{b}$ is an inner bucket (i.e., all buckets are outer buckets)}
If all components of $\mathbf{u}$ fall in either their first or last buckets, then the network is uniquely defined. This occurs because agent-level preferences for forming (or not forming) a link are so strong that they do not depend on the presence or absence of other links in the network. Each agent i either prefers to send a link to j, regardless of the actions taken by others, or does not wish to send a link. Put differently, each agent has a pure link formation strategy which is strictly dominating in such games; therefore $\mathcal{N}\left(\mathbf{d},\mathbf{u};\theta\right)$ is either zero or one.

For a particular network $\mathbf{d}$, $\mathcal{N}\left(\mathbf{d},\mathbf{u};\theta\right)=1$ if, for all (directed) dyads $ij$ such that $d_{ij}=1$, we have that $u_{ij}$ falls in the first bucket and for all dyads $ij$ such that $d_{ij}=0$ we have that $u_{ij}$ falls in the last bucket. These considerations give the equality
\begin{align}\label{eq:case_two_derivative}
\int_{\mathbf{u}\in\mathbf{b}}\mathcal{N}\left(\mathbf{d},\mathbf{u};\theta\right)f_{\mathbf{U}}\left(\mathbf{u}\right)\mathrm{d}\mathbf{u}	
    &= \prod_{i\neq j}\left[\int_{-\infty}^{\mu_{ij}+\gamma\underline{s}}f_{U}\left(u_{ij}\right)\mathrm{d}u_{ij}\right]^{d_{ij}}\left[\int_{\mu_{ij}+\gamma\bar{s}}^{\infty}    f_{U}\left(u_{ij}\right)\mathrm{d}u_{ij}\right]^{1-d_{ij}} \\
	&= \prod_{i\neq j}\left[F_{U}\left(\mu_{ij}+\gamma\underline{s}\right)\right]^{d_{ij}}\left[1-F_{U}\left(\mu_{ij}+\gamma\bar{s}\right)\right]^{1-d_{ij}}    
\end{align}
Taking logarithms of the expression above, differentiating with respect to $\gamma$, evaluating at $\gamma=0$, and multiplying by $P_{0}\left(\mathbf{d};\delta\right)$ yields a derivative for summands where all buckets in $\mathbf{b}$ are outer buckets of
\begin{equation}
P_{0}\left(\mathbf{d};\delta\right)\sum_{i\neq j}\left[d_{ij}\underline{s}\frac{f_{U}\left(\mu_{ij}\right)}{F_{U}\left(\mu_{ij}\right)}-\left(1-d_{ij}\right)\bar{s}\frac{f_{U}\left(\mu_{ij}\right)}{1-F_{U}\left(\mu_{ij}\right)}\right].
\end{equation}

\subsubsection*{Case 3: Exactly one bucket in $\mathbf{b}$ is an inner bucket}

If all but one component of $\mathbf{u}$ falls into its first or last bucket, then the resulting network is uniquely defined except for the presence or absence of one arc, say, $ij$. For any such draw of $\mathbf{u}$, since all other links are formed according to a strictly dominating strategy, player $i$ will either benefit from forming the $ij$ arc or not. Hence $\mathcal{N}\left(\mathbf{d},\mathbf{u};\theta\right)$ is also either zero or one in this case.

For a particular network $\mathbf{d}$, $\mathcal{N}\left(\mathbf{d},\mathbf{u};\theta\right)$ will equal one if two conditions hold. First, for all directed dyads $kl\neq ij$ such that $d_{kl}=1$ we have that $u_{kl}$ falls in the first bucket and for all dyads $kl\neq ij$ such that $d_{kl}=0$ we have that $u_{kl}$ falls in the last bucket. Second, for the dyad $ij$ with $u_{ij}$ falling in an inner bucket, we require that if $u_{ij}\in\left[\mu_{ij}+\gamma\underline{s},\mu_{kl}+\gamma s_{ij}\left(\mathbf{d}\right)\right)$ that $d_{ij}=1$, while if $u_{ij}=\left[\mu_{kl}+\gamma s_{ij}\left(\mathbf{d}\right),\mu_{ij}+\gamma\bar{s}\right)$ we require that $d_{ij}=0$. The overall likelihood contribution for this case therefore equals:
\begin{align*}
\int_{\mathbf{u}\in\mathbf{b}}\mathcal{N}\left(\mathbf{d},\mathbf{u};\theta\right)f_{\mathbf{U}}\left(\mathbf{u}\right)\mathrm{d}\mathbf{u}
   =&	\prod_{kl\neq ij}\left[\int_{-\infty}^{\mu_{kl}+\gamma\underline{s}}f_{U}\left(u_{kl}\right)\mathrm{d}u_{kl}\right]^{d_{kl}}\left[\int_{\mu_{kl}+\gamma\bar{s}}^{\infty}f_{U}\left(u_{kl}\right)\mathrm{d}u_{kl}\right]^{1-d_{kl}} \\
	&\times\left[\int_{\mu_{ij}+\gamma\underline{s}}^{\mu_{ij}+\gamma s_{ij}\left(\mathbf{d}\right)}f_{U}\left(u_{ij}\right)\mathrm{d}u_{ij}\right]^{d_{ij}}\left[\int_{\mu_{ij}+\gamma s_{ij}\left(\mathbf{d}\right)}^{\mu_{ij}+\gamma\bar{s}}f_{U}\left(u_{ij}\right)\mathrm{d}u_{ij}\right]^{1-d_{ij}} \\
   =&	\prod_{kl\neq ij}\left[F_{U}\left(\mu_{kl}+\gamma\underline{s}\right)\right]^{d_{kl}}\left[1-F_{U}\left(\mu_{kl}+\gamma\bar{s}\right)\right]^{1-d_{kl}} \\
	&\times\left[F_{U}\left(\mu_{ij}+\gamma s_{ij}\left(\mathbf{d}\right)\right)-F_{U}\left(\mu_{ij}+\gamma\underline{s}\right)\right]^{d_{ij}} \\
	&\times\left[F_{U}\left(\mu_{ij}+\gamma\bar{s}\right)-F_{U}\left(\mu_{ij}+\gamma s_{ij}\left(\mathbf{d}\right)\right)\right]^{1-d_{ij}}.
\end{align*}
Recall that $s_{ij}\left(\mathbf{d}\right)$ is the ``strategic" part of the marginal utility agent $i$ gets if he forms $ij$. Because the last two terms in $\left[\cdot\right]$ in the expression above are zero at $\gamma=0$ we only need to consider their derivative (by the product rule the other term equals zero at $\gamma=0$). Differentiating the last two terms with respect to $\gamma$ (and multiplying by the balance of preceding terms) yields	
\begin{align*}
    &\prod_{kl\neq ij}\left[F_{U}\left(\mu_{kl}+\gamma\underline{s}\right)\right]^{d_{kl}}\left[1-F_{U}\left(\mu_{kl}+\gamma\bar{s}\right)\right]^{1-d_{kl}} \\
	&\times\left[s_{ij}\left(\mathbf{d}\right)f_{U}\left(\mu_{ij}+\gamma s_{ij}\left(\mathbf{d}\right)\right)-\underline{s}f_{U}\left(\mu_{ij}+\gamma\underline{s}\right)\right]^{d_{ij}} \\
	&\times\left[\bar{s}f_{U}\left(\mu_{ij}+\gamma\bar{s}\right)-s_{ij}\left(\mathbf{d}+ij\right)f_{U}\left(\mu_{ij}+\gamma s_{ij}\left(\mathbf{d}+ij\right)\right)\right]^{1-d_{ij}} \\
   =&	\prod_{i\neq j}\left[F_{U}\left(\mu_{ij}+\gamma\underline{s}\right)\right]^{d_{ij}}\left[1-F_{U}\left(\mu_{ij}+\gamma\bar{s}\right)\right]^{1-d_{ij}} \\
	&\times\left[s_{ij}\left(\mathbf{d}\right)\frac{f_{U}\left(\mu_{ij}+\gamma s_{ij}\left(\mathbf{d}\right)\right)}{F_{U}\left(\mu_{ij}+\gamma\underline{s}\right)}-\underline{s}\frac{f_{U}\left(\mu_{ij}+\gamma\underline{s}\right)}{F_{U}\left(\mu_{ij}+\gamma\underline{s}\right)}\right]^{d_{ij}} \\
	&\times\left[\bar{s}\frac{f_{U}\left(\mu_{ij}+\gamma\bar{s}\right)}{F_{U}\left(\mu_{ij}+\gamma\bar{s}\right)}-s_{ij}\left(\mathbf{d}\right)\frac{f_{U}\left(\mu_{ij}+\gamma s_{ij}\left(\mathbf{d}\right)\right)}{F_{U}\left(\mu_{ij}+\gamma\bar{s}\right)}\right]^{1-d_{ij}}.
\end{align*}
Summing this expression over all potential arcs (and evaluating at $\gamma=0$) gives a total contribution of “one inner bucket in $\mathbf{b}$” summands to the derivative of:
\begin{multline}\label{eq:case_three_derivative}
P_{0}\left(\mathbf{d};\delta\right)\sum_{i\neq j}d_{ij}\left[s_{ij}\left(\mathbf{d}\right)\frac{f_{U}\left(\mu_{ij}\right)}{F_{U}\left(\mu_{ij}\right)}-\underline{s}\frac{f_{U}\left(\mu_{ij}\right)}{F_{U}\left(\mu_{ij}\right)}\right] \\
+\left(1-d_{ij}\right)\left[\bar{s}\frac{f_{U}\left(\mu_{ij}\right)}{F_{U}\left(\mu_{ij}\right)}-s_{ij}\left(\mathbf{d}\right)\frac{f_{U}\left(\mu_{ij}\right)}{F_{U}\left(\mu_{ij}\right)}\right].
\end{multline}
Summing \eqref{eq:case_two_derivative} and \eqref{eq:case_three_derivative} then gives the expression in the statement of Theorem \ref{thm: derivative}. Using similar methods we can show that $P\left(\mathbf{d};\theta\right)$ can be differentiated with respect to $\gamma$ twice as claimed.

\subsection{MCMC Proofs}
\subsubsection*{Proof of Lemma \ref{lemma:symetrie}} \label{proof_symetrie}
Let $A_{\mathbf{D},\mathbf{D}'}$ be the set of arcs form the node $v_\mathbf{D}$ to the node $v_{\mathbf{D}'}$. We construct a bijection $\varphi: A_{\mathbf{D},\mathbf{D}'} \rightarrow A_{\mathbf{D}',\mathbf{D}}$. Then we show that the probability of an arc $p(a)$ is equal to $p(\varphi(a))$. If that is proven, the probability of a transition form $v_\mathbf{D}$ to $v_{\mathbf{D}'}$ is 
\begin{align*}
\sum_{a \in A_{\mathbf{D}',\mathbf{D}}} p(a) & =  \sum_{a \in A_{\mathbf{D}',\mathbf{D}}} p(\varphi (a))\\
&=  \sum_{\varphi^{-1} (a') \in A_{\mathbf{D}',\mathbf{D}}} p(a')\\
&=  \sum_{a' \in \varphi(A_{\mathbf{D}',\mathbf{D}})} p(a')\\
&=  \sum_{a' \in A_{\mathbf{D},\mathbf{D}'}} p(a')
\end{align*}
which is the probability for a transition from $v_{\mathbf{D}'}$ to $v_\mathbf{D}$.\\
For the construction of the bijection consider that every arc $A_{\mathbf{D},\mathbf{D}'}$ corresponds uniquely to a schlaufen-sequence $\mathcal{R}=(R_1,..,R_{h})$. Let $R_k = (i_1, .., i_m,..,i_l)$ with $i_m$ the start of the cycle (if there is no cycle in $R$, we set $R=\bar{R}$). We define $\bar R_k = (i_1, .., i_m, i_{l-1}, .. i_{m+1}, i_l)$ and $\bar{\mathcal{R}} = (\bar R_1,..,\bar R_{h})$.

Note that the $R_1,..,R_{h}$ are link disjoint and as soon as the cycle of $R_k$ is switched $\bar{ R}_k$ is a schlaufe. The violation matrix of $\bar R_k$ is the negative violation matrix of $R_k$. This implies that if $\mathcal{R}$ is a feasible schlaufen sequence for $G$ which defined an arc in  $A_{\mathbf{D},\mathbf{D}'}$ then $\bar{ \mathcal{R}}$ is a feasible schlaufen-sequence for $\mathbf{D}'$ and defines an arc $A_{\mathbf{D}',\mathbf{D}}$.

We define now $\varphi$ as the function which maps the arc in $ A_{\mathbf{D},\mathbf{D}'}$ with schlaufen sequence  $\mathcal{R}$ to the arc in $A_{\mathbf{D}',\mathbf{D}}$ with schlaufen sequence $ \bar{\mathcal{R}}$. By construction $\varphi$ is injective, which implies  $ |A_{\mathbf{D},\mathbf{D}'}| \leq | A_{\mathbf{D}',\mathbf{D}}|$. By symmetry we conclude  $ |A_{\mathbf{D}',\mathbf{D}}| \geq | A_{\mathbf{D},\mathbf{D}'}|$, which implies $ |A_{\mathbf{D}',\mathbf{D}}| = | A_{\mathbf{D},\mathbf{D}'}|$ and that $\varphi$ is bijective.

It remains to show that the probability of an arc $p(a)$ is equal to $p(\varphi(a))$. For any node there are equally many feasible active / passive outlinks in $\mathbf{D}$ as in $\mathbf{D}'$. If for a node one outlink is marked due to an link in $R_k$ then for the same node one outlink is marked in $\bar{R}_k$. Therefore $r_{G_k'}(i)$ is equal to  $r_{G_k}(i)$ for an active as well as a passive step. Looking at equation \eqref{eq:schlaufen} the $p_{G_k}(R_k)$ is only different from $p_{G_k'}(\bar{R}_k)$ in the numbering of the factors. But in a cycle of a schlaufe the start node $i_m$ and the end node $i_l$ are such that $m-l \mod 2 = 0$. The reordering leaves even indexes even and odd indexes odd. Therefore $p_G(R_k)$ = $p_G'(\bar{R}_k)$. From equation \eqref{eq:prob} it follows directly that $p_G(\mathcal{R})=p_{G'}(\bar{\mathcal{R}})$ which completes the proof.

\subsubsection*{Proof of Lemma \ref{lemma:connected}}\label{connected}

The symmetric difference of two realizations of $\mathbb{D}_{s,m}$, which we denote by $\mathbf{D}$ and $\mathbf{D}'$ is a set of alternating cycles. Cycles are in particular schlaufen. We order them arbitrarily as $(R_1, ..,R_h)$.
The sum of the violation matrices is 0. Therefore $(R_1, ..,R_h)$ is either a feasible schlaufen-sequence or a concatenation of feasible schlaufen-sequences. In the first case there is an arc from $v_\mathbf{D}$ to $v_{\mathbf{D}'}$. 
In the second case, all the feasible schlaufen-sequence define an arc to a new node, resulting in a directed path starting at $v_\mathbf{D}$ and ending in $v_{\mathbf{D}'}$. Thus between any two vertexes in $\Phi$ there is a directed path.

\subsubsection*{Proof of Theorem \ref{theorem:correctness}}\label{correctness1}
Every time the Algorithm \ref{alg: markov_draw} arrives at step 2 a new arc of the state graph is crossed. At step 2 the algorithm follows a loop arc of type 1 with probability $q$. Otherwise it proceeds to step 3. In step 3 a schlaufen-sequence $\mathcal{R}$ is constructed. If the violation matrices of this schlaufen-sequence sum up to 0, the its cycles are switched and an arc of type 2 is followed with probability $p_G(\mathcal{R}$). If the violation matrices do not sum up to 0, then an arc of type 3 is followed. All the cases in which the violation matrices don't sum up to 0 correspond to the residual probability. Therefore Algorithm \ref{alg: markov_draw} is a a random walk on the state graph $\Phi$.

According to Lemma \ref{lemma:symetrie} $\Phi$ is (weighted) symmetric and according to Lemma \ref{lemma:connected} it is strongly connected. Due to the self-loops, $\Phi$ is not bipartite. Therefore the limit distribution is uniform.

\subsection{Monte Carlo experiments}\label{app: monte_carlo}
In this appendix we summarize the results of a small number of Monte Carlo experiments. These experiments illustrate the excellent size control and good power properties of our tests. For the Monte Carlo experiments we work with the general utility function
introduced in the main paper. We assume that $A_{i}\in\mathbb{A}\overset{def}{\equiv}\left\{ \alpha_{L},\alpha_{H}\right\} $,
$B_{i}\in\mathbb{B}\overset{def}{\equiv}\left\{ \beta_{L},\beta_{H}\right\}$ and $X_{i}\in\mathbb{X}\overset{def}{\equiv}\left\{ 0,1\right\}$. We
assume that each support point in $\mathbb{A}\times\mathbb{B}\times\mathbb{X}$ occurs with equal probability (i.e., with probability equal to $\frac{1}{8}$).

Observe that their are four types of sending agents: ($A_{i}=\alpha_{L},X_{i}=0$), ($A_{i}=\alpha_{H},X_{i}=0$), ($A_{i}=\alpha_{L},X_{i}=1$) and $(A_{i}=\alpha_{H},X_{i}=1)$.
Similarly there are four types of receiving agents. The null model is therefore fully described by $16=4\times4$ linking probabilities. These probabilities are, in turn, a function of the $8$ model parameters. We set these parameters as follows: $-\alpha_{L}=\alpha_{H} =0.7$,  $-\beta_{L}=\beta_{H}=0.5$, $\lambda_{00}=\lambda_{11}=-2$ and $\lambda_{10}=\lambda_{01}=-4$. This yields a null model expected graph density of about 0.10.

Our parameter choices generate meaningful degree heterogeneity and homophily under the null. Across $1,000$ Monte Carlo simulations with $N=100$, average network density was 0.1, average reciprocity was 0.1, and the \emph{average} standard deviation of in- and out-degree, was 7.5.

We set the network benefit function to $g_{i}\left(\mathbf{d}\right)=\sum_{j}d_{ij}\left(\sum_{k}d_{ik}d_{kj}\right)$ as is appropriate when agents prefer transitive ties. To simulate a network under the alternative we draw $\mathbf{U}$ and then, starting with an empty adjacency matrix, iterate until all links with positive marginal utility are present, and all those with negative marginal utility are absent. By Tarski's Theorem this finds us the least dense pure strategy Nash Equilibrium.

We compare the performance of three tests: (i) the infeasible locally best test that is based upon the true value of $\delta = \delta_{0}$; (ii) the
feasible version of this test which replaces $\delta$ with its maximum likelihood estimate computed under the null; finally, (iii) we construct an
\emph{ad hoc} test based upon the transitivity index. This last test is the one most often used in practice (where it is compared to a reference value derived from a simple random graph null).

Figure \ref{fig: monte_carlo_results} summarizes our findings. The horizontal axis of the figure correspond to different values of the strategic interaction parameter, $\gamma_{0}$; the vertical axis to the rejection frequency. With 1000 Monte Carlo replications the standard error of our simulation estimate of size is $\sqrt{(0.05\left(1-0.05\right)/1000)}\approx0.007$.

\begin{figure}[htp]
    \caption{Power Analysis}
    \centering
    {{\includegraphics[width=15cm]{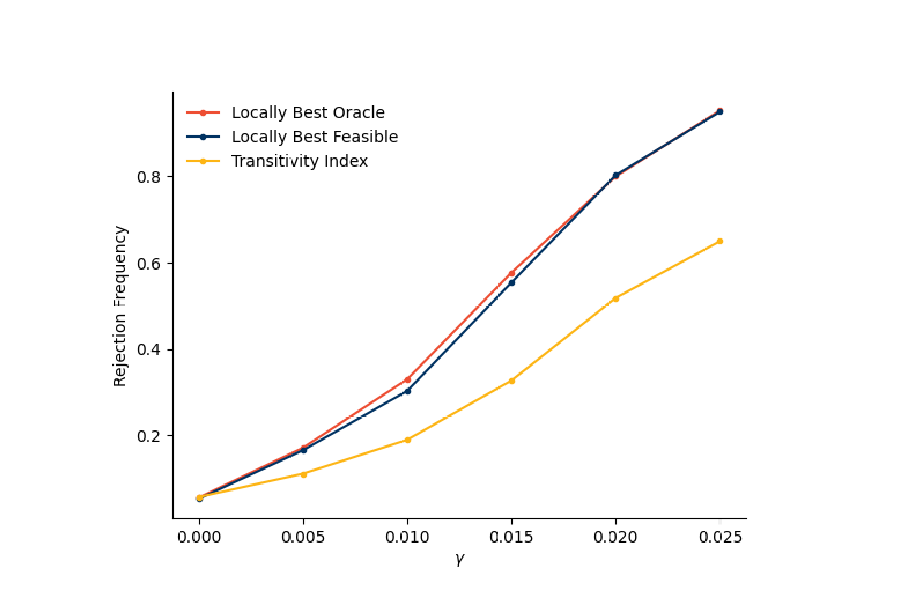} }}%
    \caption*{\underline{Source:} Authors' calculations. \hfill\break
    \underline{Notes:}  The figures plot the frequency with which $H_{0}:\gamma_{0}=0$ is rejected across $1,000$ Monte Carlo replications for networks with $N = 100$ agents. The y-axis reports the estimated rejection frequency, the x-axis gives the value of the strategic interaction parameter, $\gamma$. The minimal pure strategy NE is use to simulate each network. For each simulation a total of 100 MCMC draws from $\mathbb{D}_{\mathbf{s},\mathbf{m}}$ were used to compute critical values. The mixing time was chosen such that (approximately) every edge is twice randomly modified before the network is considered a uniform draw. The marginal utility function equals $A_{i}+B_{j}+X_{i}'\Lambda_{0}X_{j}+\gamma_{0}s_{ij}(\mathbf{d})-U_{ij}$ with $s_{ij}\left(\mathbf{d}\right)=\sum_{k}d_{ik}d_{kj}+\sum_{k\neq j}d_{ik}d_{jk}$. The distribution of $\left(A_{i},B_{i},X_{i}\right)$ and the model parameters are as described in the main text.}
    \label{fig: monte_carlo_results}
\end{figure}

As expected, the actual size of our test is indistinguishable (i.e., equal up to simulation error) from its nominal size. For the designs considered here the power gains associated with using the locally best test statistic derived in Section \ref{sec:test} are considerable. Furthermore the feasible locally best test, which replaces $\delta_0$ with its MLE (computed under the null), performs almost as well as the infeasible locally best test based on the actual value of $\delta_0$.

The Monte Carlo experiments highlight that the locally best test, which upweights episodes of ``unexpected" transitivity, is more powerful than the ad hoc test based on comparing the transitivity index with its null distribution. Note both tests are valid and correctly-sized.

Next we consider the behavior of our test under mis-specification. Specifically we consider a data generating process where the link-specific random utility shocks are Gaussian instead of logistic. We set the variance of the Gaussian distribution to $\pi^2/3$ so that their scale is the same as in the logistic case. All other model parameters remain as defined above. We simulate $1,000$ networks with Gaussian errors, but then proceed ``as if" our logistic assumption were true. Note our sufficiency, conditioning, and similarity arguments are no longer valid.

The results of these experiments are summarized in Figure \ref{fig: MC_Gaussian}. Our tests are conservative in these experiments, with actual sizes below their nominal $0.05$ level. The power properties of our tests remain good. Note the ``oracle" test is based on an evaluation of a logit probability function at the true utility function coefficients (from the the Gaussian model). The feasible test is based upon quasi maximum likelihood estimates of the pseudo-true utility function coefficients (under the null). There is no a priori reason to expect one test to perform better than the other in these designs (indeed one might expect the feasible test to do better under misspecification). In this example the two power curves cross.

This is just one experiment, but it provides some suggestive evidence that our test may still be useful in settings with modest departures from the logistic assumption.

\begin{figure}[htp]
    \caption{Test Behavior Under a Gaussian Random Utility Distribution}
    \centering
    \center{}
    {{\includegraphics[width=15cm]{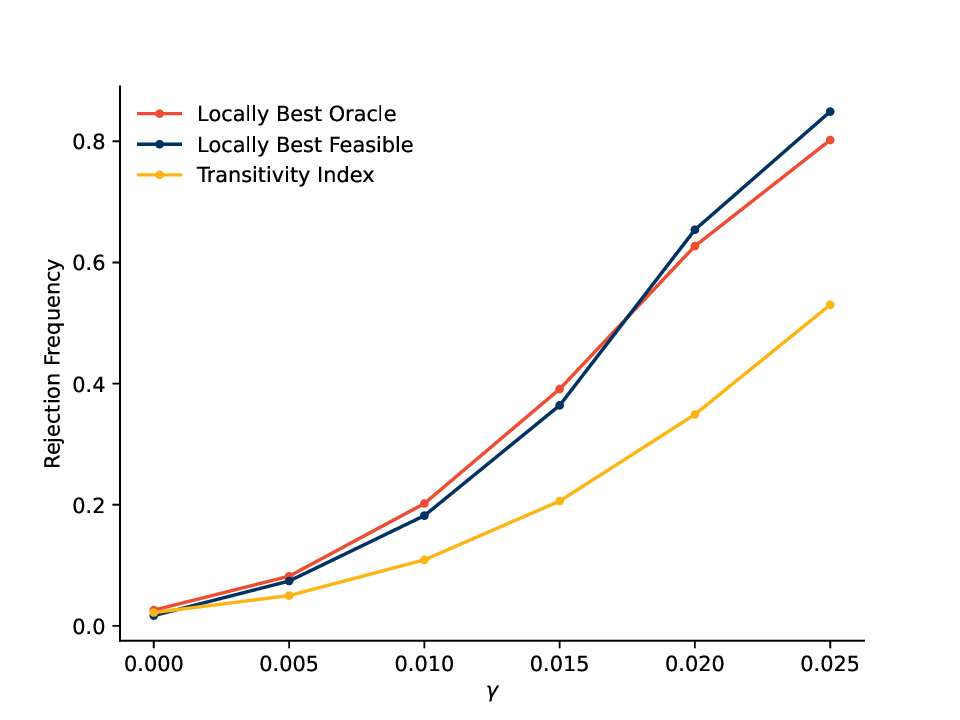} }}%
    \caption*{\underline{Source:} Authors' calculations. \hfill\break
              \underline{Notes:} All features of these experiments are as reported in the notes to Figure \ref{fig: monte_carlo_results} above, with the exception that the link-specific random utility shocks are Gaussian (with a variance of $\pi^2/3$). Estimation proceeds ``as if" the logistic assumption nevertheless holds.}
    \label{fig: MC_Gaussian}              
\end{figure}

\subsection{Additional applications}\label{app: applications}
\subsubsection*{Bi-partite networks}

Let $i=1,\ldots,N_{1}$ index a set of firms deciding which of $j=1,\ldots,N_{2}$ markets to enter or not. To be concrete consider the problem of airlines deciding which routes they will operate in. The $K_{1}\times1$ binary vector $X_{i}^{f}$ indicates what type firm $i$ is (e.g., a legacy carrier or a low cost airline); the binary vector $X_{j}^{m}$ indicates what type market $j$ is (e.g., small, medium or large). Let $\mathbf{E}=\left[E_{ij}\right]_{i=1,\ldots,N_{1},j=1,\ldots,N_{2}}$ be the $N_{1}\times N_{2}$ matrix which records which markets each firm chooses to enter (“$\mathbf{E}$” for market entry matrix).

The payoff firm $i=1,\ldots,N_{1}$ gets from a given industry-wide pattern of market entry $\mathbf{E}=\mathbf{e}$ equals
\begin{equation}
    \nu_{i}\left(\mathbf{e}_{i},\mathbf{e}_{-i};\theta,\mathbf{U}_{i}\right)=\gamma_{0}g_{i}\left(\mathbf{e}\right)+\sum_{j=1}^{N_{2}}e_{ij}\left(A_{i}+B_{j}+X_{i}^{f\prime}\Delta X_{j}^{m}-U_{ij}\right).
\end{equation}
Here $A_{i}$ captures unobserved heterogeneity across firms; for example some firms may be able to operate at systematically lower cost and hence profitably enter more markets (corresponding to high values of $A_{i}$). Likewise $B_{j}$ captures heterogeneity across markets; some markets may be intrinsically more profitable than others and hence many firms operate in them (corresponding to high values of $B_{j}$). The term $X_{i}^{f\prime}\Delta X_{j}^{m}$ allows certain types of markets to be systematically more attractive to certain types of firms. Finally $U_{ij}$ is a firm-by-market specific logistic profit shock.

In a simple entry game, we might set $g_{i}\left(\mathbf{e}\right)=-\sum_{j=1}^{N_{2}}e_{ij}\left[\sum_{k=1}^{N_{1}}e_{kj}\right]$, such that entry into market $j$ is less attractive to firm $i$ when many other firms $k$ also enter market $j$ \citep[e.g.,][]{Ciliberto_Tamer_EM2009}. This payoff function implies independence of entry decisions across markets.

An advantage of thinking about entry decisions as resulting in a bi-partite network, with arcs from firms to markets, is that it makes it easy to consider more complex preference structures. These preference structures can allow for interdependence in entry decisions across markets.

For example, in the model of \cite{Jia_EM08} entry into market $j$ is more attractive if the firm also enters other nearby markets. In this case we might set
\begin{equation}
    g_{i}\left(\mathbf{e}\right)=-\sum_{j=1}^{N_{2}}e_{ij}\left[\sum_{k=1}^{N_{1}}e_{kj}\right] + \lambda\left(\mathbf{e}_{i}\right)
\end{equation}
where $\lambda\left(\mathbf{e}_{i}\right)$ varies inversely with some measure of the spatial spread of those markets $i$ enters (calibrated to measure how the operating costs of the firm vary with the geographic dispersion of the markets entered). In the airline example it may be more costly, or less profitable, to operate across markets that are disconnected. That is an airline may prefer route maps which allow a customer to travel to all other airports in their network without having to fly with competitor. In this example $\lambda\left(\mathbf{e}_{i}\right)$ might equal a decreasing function of the number of connected components in $i$'s route network. Alternatively it could vary with the number of location pairs which can be reached on either a direct flight or by making just one connection.

We can also allow for more complex forms of spatial competition. If $j$ indexes airline routes, then we might set
\begin{equation}
    g_{i}\left(\mathbf{e}\right)=-\sum_{j=1}^{N_{2}}e_{ij}\phi_{j}\left(e_{-i}\right)
\end{equation}
with $\phi_{j}\left(\mathbf{e}_{-i}\right)$ returning how many competitor airlines operate in routes with origin and destination airports both within a one hour drive of the corresponding route $j$ airports (e.g., entry into the SFO-LAX market may depend on how many competitors operate on the OAK-BUR, SFO-BUR, and OAK-LAX routes as well as the number which operate on the SFO-LAX route). One apparent advantage of this “network” perspective is that it allows for the incorporation of complex cross-market complementarities as well as rich forms of spatial competition.

To connect this entry model to the directed network problem defined in the main text we define the $N\times N$ matrix:
\begin{equation}
    \mathbf{D}=\left[\begin{array}{cc}
                        \mathbf{0} & \mathbf{E}\\
                        \mathbf{0} & \mathbf{0}
                    \end{array}\right]
\end{equation}
and then proceed as described in the paper.

The network $\mathbf{D}$ is bipartite. There are no firm-to-firm or market-to-market links. Furthermore, only firms may direct links (with arcs denoting market entry decisions). These features of the problem induce the special structure of  $\mathbf{D}$ above. The cross link matrix will also have a structure analogous to that of the adjacency matrix.

In this example the null reference distribution will assign a zero to many ``links" with probability one; this is not a problem for our MCMC simulation algorithm. The implementation available in the Python \textbf{ugd} package can handle degree sequences and cross link matrices with zero elements.

The null set of networks corresponds to all networks where the same number of firms enter each market as observed in the network in hand and the observed patterns of entry are the same as in the network in hand. For example, the number of low cost airlines entering small, medium and large markets is held fixed and so on.

A natural test statistic would be
\begin{equation}
    R\left(\mathbf{D}\right)=\sum_{i=1}^{N_{1}}\sum_{j=1}^{N_{2}}\left(E_{ij}-p_{ij}\left(\hat{\delta}\right)\right)s_{ij}\left(E_{ij}\right)    
\end{equation}
with $p_{ij}\left(\hat{\delta}\right)$ an estimate of the probability that firm $i$ enters market $j$ under the null that $\gamma_{0}=0$ and $s_{ij}\left(\mathbf{E}\right)$ the marginal network benefit of entering market $j$ for firm $i$ holding all other own and competitor route choices fixed. 

As the above example makes clear, the application of the methods proposed in the main text to bipartite networks is conceptionally straightforward. This, in turn, expands the class of many player games to which our methods apply. Note that questions of test power are game specific.

\subsubsection*{Conditional inference in dyadic models}

Following a suggestion in \cite{Graham_EM17}, our MCMC algorithm can also be used for conditional maximum likelihood estimation (CMLE) of non-strategic dyadic logit models. Let $D_{ij}=1$ if country $i$ attacks country $j$ within some researcher-defined time period. Let $X_{i}$ be a $K \times 1$ vector of country types (e.g., a partition of countries into broad geographic regions), finally let $Z_{ij}=1$ if both $i$ and $j$ are democracies and zero otherwise. We posit the following model for the initiation of conflict by $i$ against $j$
\begin{equation}
    D_{ij}=\mathbf{1}\left(A_{i}+B_{j}+X_{i}'\Lambda_{0}X_{j}+Z_{ij}'\beta_{0}-U_{ij}\geq0\right),
\end{equation}
for $i\neq j$, $i,j=1,\ldots,N$ and $U_{ij}$ logistic.

According to democratic peace theory, democracies are less likely to engage in conflict with other democracies such that $\beta_{0}<0$ \citep[e.g.][]{Oneal_Russett_WP99}. Note that any level (or monadic) effect of democracy on the propensity to initiate conflict generally is absorbed into the ego effects $\left\{ A_{i}\right\} _{i=1}^{N}$ (out-degree effects), while any level effect on the propensity to be militarily targeted by others is absorbed into the alter effects $\left\{ B_{j}\right\} _{j=1}^{N}$ (in-degree effects). Systematic cross-regional patterns in the costs and benefits of conflict are controlled for by the “homophily” term $X_{i}'\Lambda_{0}X_{j}=W_{ij}'\lambda_{0}$.

The conditional likelihood of the network in hand, $\mathbf{D}$, here the observed pattern of conflict among nations, is
\begin{align*}
        L\left(\mathbf{D};\delta_{0},\beta_{0}\right) 
        =& \prod_{i\neq         j}\left[\frac{\exp\left(W_{ij}'\lambda_{0}+R_{i}'\mathbf{A}+R_{j}'\mathbf{B}+Z_{ij}'\beta_{0}\right)}{1+\exp\left(W_{ij}'\lambda_{0}+R_{i}'\mathbf{A}+R_{j}'\mathbf{B}+Z_{ij}'\beta_{0}\right)}\right]^{D_{ij}}\\
	    & \times\left[\frac{1}{1+\exp\left(W_{ij}'\lambda_{0}+R_{i}'\mathbf{A}+R_{j}'    \mathbf{B}+Z_{ij}'\beta_{0}\right)}\right]^{1-D_{ij}}\\
        =& c\left(\mathbf{X},\mathbf{Z};\delta_{0},\beta_{0}\right)\prod_{i\neq     j}\exp\left(\mathbf{T}'\delta_{0}\right)\exp\left(\left[\sum_{i\neq j}D_{ij}Z_{ij}\right]'\beta_{0}\right)   
\end{align*}
with $R_{i}$, as earlier, a $N \times 1$ vector with a $1$ in its $i^{th}$ element and zeros elsewhere and $c\left(\mathbf{X},\mathbf{Z};\delta_{0},\beta_{0}\right)$ not varying with $\mathbf{D}$. Here $\mathbf{T}$ includes the vectorized cross-link matrix as well as the out- and in-degree sequences as discussed in the main text.

Let $\mathbf{\tau}\left(\mathbf{v}\right)$ be the sufficient statistics for $\delta$ in network $\mathbf{v}$. Conditioning on $\mathbf{T=t}$ yields
\begin{align*}
        \Pr\left(\left.\mathbf{D}=\mathbf{d}\right|\mathbf{X}= 
        \mathbf{x},\mathbf{Z}=\mathbf{z},\mathbf{T}=\mathbf{t};\delta_{0},\beta_{0}\right)
        =& \frac{\prod_{i\neq j}\exp\left(\mathbf{t}'\delta_{0}\right)\exp\left(\left[\sum_{i\neq j}d_{ij}z_{ij}\right]'\beta_{0}\right)}{\sum_{\mathbf{v}\in\mathbb{D}_{\mathbf{s},\mathbf{m}}}\prod_{i\neq j}\exp\left(\mathbf{\tau}\left(\mathbf{v}\right)'\delta_{0}\right)\exp\left(\left[\sum_{i\neq j}v_{ij}z_{ij}\right]'\beta_{0}\right)}\\
	    =& \frac{\exp\left(\left[\sum_{i\neq j}d_{ij}z_{ij}\right]'\beta_{0}\right)}{\sum_{\mathbf{v}\in\mathbb{D}_{\mathbf{s},\mathbf{m}}}\exp\left(\left[\sum_{i\neq j}v_{ij}z_{ij}\right]'\beta_{0}\right)}\\
	    =& \Pr\left(\left.\mathbf{D}=\mathbf{d}\right|\mathbf{Z}=\mathbf{z},\mathbf{T}=\mathbf{t},\beta_{0}\right),
\end{align*}
where the second equality follows from the fact that $\mathbf{\tau}\left(\mathbf{v}\right)=\mathbf{t}$ for all $\mathbf{v}$ in $\mathbb{D}_{\mathbf{s},\mathbf{m}}$.

Taking logs yields, after some manipulation,
\begin{multline}
    \ln\Pr\left(\left.\mathbf{D}=\mathbf{d}\right|\mathbf{Z}=\mathbf{z},\mathbf{T}=\mathbf{t},\beta_{0}\right)	= \left[\sum_{i\neq j}d_{ij}z_{ij}\right]'\beta_{0}\\
    -\ln\left(\mathbb{E}\left[\exp\left(\left[\sum_{i\neq j}D_{ij}z_{ij}\right]'\beta_{0}\right)\right]\right)-\ln\left|\mathbb{D}_{\mathbf{s},\mathbf{m}}\right|,\label{eq: conditional_dyadic_likelihood}
\end{multline} 
where the expectation in the second term to the right of the equality is taken with respect to the discrete uniform distribution on $\mathbb{D}_{\mathbf{s},\mathbf{m}}$. This expectation can be estimated using our simulation algorithm. The third term is invariant to $\beta_{0}$ and can consequently be ignored. Basing estimation and inference upon \eqref{eq: conditional_dyadic_likelihood} allows a researcher to learn about $\beta_{0}$ in the presence of out- and in-degree heterogeneity as well as potentially complex patterns of homophily.

In our democratic peace theory example, $Z_{ij}$ is binary. In this case $\sum_{i\neq j}D_{ij}Z_{ij}$ is simply a count of how many wars are initiated by democracies against other democracies. To compute the expectation in \eqref{eq: conditional_dyadic_likelihood} we need an estimate of the distribution of this count induced by the discrete uniform distribution on $\mathbb{D}_{\mathbf{s},\mathbf{m}}$.

\end{document}